\documentclass[twoside]{amsart}
\usepackage{graphicx}
\usepackage{amsthm}
\usepackage{amsmath}
\usepackage{amssymb}
\usepackage{color}
\usepackage[T1]{fontenc}
\usepackage{dsfont}
\usepackage[normalem]{ulem}

\newtheorem{thm}{Theorem}[section]

\newtheorem{cor}[thm]{Corollary}
\newtheorem{lem}[thm]{Lemma}
\newtheorem{prop}[thm]{Proposition}
\theoremstyle{definition}
\newtheorem{defn}[thm]{Definition}
\newtheorem{example}{Example}

\theoremstyle{remark}
\newtheorem{rem}[thm]{Remark}

 \numberwithin{equation}{section}
\newcommand{\norm}[1]{\left\Vert#1\right\Vert}

\newcommand{\set}[1]{{\left\{#1\right\}}}

\newcommand{\I}{ \mathds{1} }
\newcommand{\Tr}{ \text{Tr} }
\newcommand{\diag}{ \text{diag} }

\newcommand{\mn}[1]{\begin{color}[rgb]{0.00, 0.0, 0.98}#1\end{color}}

\title
[Risk-minimization]
{
Risk-minimization and hedging claims on a jump-diffusion market model, Feynman-Kac Theorem and PIDE.
}
\author{Jacek Jakubowski}
\address{
Institute of Mathematics, University
of Warsaw, \\
 ul. Banacha 2, 02-097 Warszawa, Poland \\
 and \\
 Faculty of
Mathematics and
Information Science,  \\
Warsaw University of
Technology \\
E-mail: {\tt jakub@mimuw.edu.pl }
}

\author{Mariusz Niew\k{e}g\l owski}
\address{
Faculty of Mathematics and
Information Science, Warsaw University of Technology,  \\
ul. Plac Politechniki 1, 00-661 Warszawa, Poland \\
E-mail:  {\tt
M.Nieweglowski@mini.pw.edu.pl}
}
\usepackage{fancyhdr}
\usepackage{lastpage}
\begin{document}
\begin{abstract}
At first, we solve a problem of finding a risk-minimizing hedging strategy on a general market with ratings. Next,  we find a solution to this problem on Markovian market  with ratings on which prices are influenced by additional factors and rating, and behavior of this system is described by SDE driven by Wiener process and compensated Poisson random measure and claims depend on rating. To find a tool to calculate hedging strategy we prove a Feynman-Kac type theorem. This result is of independent interest and has many applications, since it enables to calculate some conditional expectations using related PIDE's.
We illustrate our theory on two examples of market. The first is  a general exponential L\'{e}vy model with stochastic volatility,  and the second is a generalization of exponential L\'{e}vy model with regime-switching.
\end{abstract}

\maketitle

\begin{center}
This version \today
\end{center}
\begin{quote}
\noindent  \textbf{Key words}: Feynman-Kac theorem, PIDE,  risk-minimization,  quadratic hedging, jump-diffusion market model, dividend process, Markovian market driven by L\'{e}vy noise.

\noindent
\textbf{AMS Subject Classification}:
Primary 60H30, 91G40, Secondary 91G80, 91G20.
\end{quote}

\section{Introduction}


F\"{o}lmer and Sondermann \cite{folsch1991}   introduced the
concept of risk-minimization as a tool for hedging and pricing in
the incomplete markets. They considered the
problem of hedging a payment at maturity
which is non-attainable, that is a claim $H$ for which there is
no self-financing strategy that replicates $H$. The idea
of F\"{o}lmer and Sondermann was to drop the self-financing condition and look for strategies
that hedges $H$ perfectly and minimize conditional variance of the
remaining cost at each time $t$. In \cite{folson1986} they proved that
there exists a unique risk-minimizing hedging strategy for an arbitrary
square integrable payoff at fixed maturity $T < \infty$ provided that
the process of discounted price of risky asset is a martingale.
 The problem of finding the explicit form of risk-minimization strategy for an European contingent claim being a function of asset price at time $T$ were solved by
 Elliott and  F\"{o}llmer \cite{ef1991}  in the one-dimensional  Markovian  case as well as in general case by using orthogonal martingale representation.
 M{\o}ller \cite{mol2001} generalized  results of F\"{o}lmer and Sondermann to the case in which
liabilities of the hedger are described by an arbitrary square integrable and c\`{a}dl\`{a}g payment process.
Subsequently, Schweizer \cite{sch1993} proved
that if a discounted
price process is not a martingale then one cannot, in general, find
a risk-minimizing strategy. To overcome this problem Schweizer
\cite{sch1993} introduced the new concept of the
local-risk-minimization which generalizes risk-minimization and
it is suitable  for general special semimartingales.
 However, to apply a local-risk-minimization strategy it is necessary to know drift coefficient, and usually it is hard to estimate it (see discussion in Tankov \cite{tankov2011}).
In this paper, we consider the risk-minimization problem.


In Section 2 we present the state of knowledge, so we recall main definitions and
 results
on existence and uniqueness of risk-minimizing  strategy for discounted cumulated payments.

 In Section 3 we solve the risk-minimization problem on  a general market with ratings. We consider a market with time horizon $T^* < \infty$ and with
$d$ primary risky assets with the price process denoted by $S$ and a money
account with the price process $B$.
The dynamics of discounted price process  $S^*_t :=  S_t / B_t$ is given by
SDE driven by a Wiener process,  a compensated   integer-valued random measure, and a martingales  $M^{i,j}$ influenced by a credit rating of some corporate or a state of
economy.
  Theorem \ref{thm:risk-min-spec} 
  states that the risk-minimizing strategy,
under appropriate conditions, solves some linear system of equations. Moreover,  using the Galtchouk-Kunita-Watanabe  decomposition of discounted cumulated payments,  we give an explicit form of this strategy in terms of coefficients which appear in the dynamics describing prices  and in terms of processes which appear in a martingale representation of discounted cumulated payments.
Existence of  this representation is  assured by a weak property of predictable representation.
There are many models in which this property holds.
 In Theorem \ref{prop:attain} we formulate the necessary and sufficient conditions for attainability of a dividend process. Since these conditions are very restrictive this theorem underline the necessity of using another methods of hedging the risk for non-attainable  dividend streams, e.g., finding of  the risk-minimizing strategies as we do here.

 In  Section 4 we will consider very flexible Markovian market model with ratings, which is  driven by
 a standard  Wiener process,  a  compensated Poisson random measure, and
 counting point processes with intensities.
 As we can see in Section 5 this
allows us to derive more explicit formulae for risk-minimizing
strategies by means of solutions to some partial integro-differential equations (PIDE's).
  We assume that information available to the market participants is modeled by a multidimensional process $(Y,C)$ given as a solution to SDE,
 where $Y=(S, R)$ with $S$ being a process of prices of tradable risky assets and $R$ represent some factors of  economic environment such as interest rates, inflation or stochastic volatility.
 The evolution
of price depends on economic conditions of market which
are described by rating system $C$ and additional non-tradable risk factor process $R$. So we assume that  dynamics of the  process $(Y,C)$ is described by SDE  \eqref{eq:SDE-gen} in which the credit rating of corporate has influence on asset prices $S$ and other non-tradable factor $R$ by
changing drift and volatility and moreover a change in credit rating
  causes a jump.
 The  money account on this market  is given by \eqref{eq:bank-acc}.
  In Theorem \ref{thm:risk-min-markov} we establish a 0-achieving risk-minimizing strategy for a general payment stream $D$ given by \eqref{eq:D} which corresponds to rating sensitive claims considered in Jakubowski and Nieweglowski \cite{jaknie2011}.
We connect this problem with finding the ex-dividend price function $v$, and
  for a sufficiently regular function $v$ we can find a  0-achieving risk-minimizing strategy for $D$ written explicitly in the terms of this function $v$.

In Section 5 we formulate and prove a Feynman-Kac type theorem for components   of a weak solution to SDE \eqref{eq:SDE-gen}. So, we connect a problem of calculation of conditional expectation of some natural functionals of these components with solution to corresponding PIDE with given boundary conditions.
Note also that if there is no component $C$ we have obtained the clasical Feynman-Kac theorem.
It is worth to notice that we solve a general problem which is of independent interest and has many applications, among others in finding the risk-minimizing strategies.

In Section 6 we discuss possibly  extensions of our results to risk-minimization problems. Moreover, we present examples of applications  of our results considering a general exponential L\'{e}vy model with stochastic volatility, a model which generalizes a regime switching model with jumps and semi-Markovian regime switching models.
In appendix we prove the result which  connects the existence  of replication strategy for a payment stream  and the form of Galtchouk-Kunita-Watanabe  decomposition of some random variable.

Our models are very general, and can be applied to the most models which appear in finance, e.g.
stochastic volatility models, regime switching Black Scholes type  models considered among  others by  Elliott, Chan and Siu \cite{ecs2005},
 Siu and Yang \cite{sy2009},
 Yao, Zhang and Zhou \cite{yzz2006}, Yuen and Yang \cite{yy2009} and also
regime switching L\'evy models studied by
 Chourdakis \cite{ch2005}, Mijatovic and Pistorius \cite{mijpis2011},
Kim, Fabozzi, Lin and Rachev \cite{kflr2012}. In our models we consider the additional factors which influence the prices, so known the results can be generalized. This justifies the assertion that our models well described a real market.

Looking at our model from point of view of insurance (so changing interpretation of abstract objects in our model) we obtain new results in  insurance, e.g. on risk minimizing hedging strategies of insurance payment processes studied by M{\o}ller \cite{mol2001}.


%

\section{Risk-minimization}
In this
section we briefly recall main definitions and theorems which allows
us to define precisely notion of risk-minimality and present results
on existence and uniqueness of risk-minimizing hedging strategies.
We will consider processes on a complete probability space
 $(\Omega, \mathcal{F}, \mathbb{P})$ with filtration $\mathbb{F} = (\mathcal{F}_t)_{t \in [\![0,T^*]\!]} $ satisfying the "usual conditions" and $\mathcal{F}=\mathcal{F}_{T^*}$.
By $\mathcal{ P }$ we denote the predictable sigma algebra on $\Omega \times  \mathbf{R}_+$.
 We consider a market with a risky
assets price process denoted by $S$ and  a bank account price process  $B$.  By $\varphi= (\phi,\eta)$ we denote
 a strategy that describes number of assets hold in
portfolio at time $t$, i.e.  number of assets invested in risky assets and in bank account, respectively.
 We specify later the measurability requirements imposed on $\varphi$.
The process $V_t(\varphi) := \phi_t^\top S_t + \eta_t
B_t$ denotes the wealth of portfolio $\varphi$ at time $t$. By
$V^*(\varphi)$ we denote the discounted wealth of the portfolio
$\varphi$, so $V^*_t(\varphi) = \phi_t^\top S^*_t + \eta_t $. We
make an assumption that $S^* :=  S_t / B_t$ is a square integrable
martingale under the measure $\mathbb{P}$. By $\langle S^*
\rangle$ we denote the matrix of predictable covariations of $S^*$,
i.e.,
\[
\langle S^* \rangle_t  = \left[
\langle S^*_i ,
S^*_j\rangle_t \right]_{i,j},
\]
where $\langle S^*_i , S^*_j\rangle_t$ is the unique predictable
proces such that $S^*_i S^*_j - \langle S^*_i , S^*_j\rangle_t$ is
a martingale. In the sequel we use convention from stochastic integration theory writing  $\int_t^u$ instead of  $\int_{]\!]t,u]\!]}$ in stochastic and Lebesgue integrals.
\begin{defn}
We say that $\varphi = (\phi,\eta)$ is an $L^2$-strategy if $\phi$
is a predictable process such that
\begin{equation}\label{eq:S-integr-cond}
 \mathbb{E} \left( \int_{0}^T \phi^{\!\top}_t d \langle S^* \rangle_t \phi_t \right) < \infty,
\end{equation}
 $\eta$ is an $\mathbb{F}$-adapted process, $V^*_t (\varphi) :=
\phi_t^\top \! S^*_t - \eta_t$ is a right continuous and square
integrable. An $L^2$-strategy $\varphi$  is
called \emph{$0$-achieving strategy}, if  $V_T (\varphi) = 0 $.
\end{defn}
We consider a contract between two parties, a seller (also
called hedger) and a buyer, which
specifies precisely the
cash-flows between these two parties. These  cash-flows are described by a
c\'{a}dl\'{a}g processes $D$, i.e. $D_t$ represents accumulated
payments (both outflows and injections of cash from the buyer) up
to time $t$. We assume additionally that $D$ matures at some
finite non-random time $T < \infty $, that is we have $D_t = D_T$
for $t \geq T$. This process $D$ is called a dividend process or a
payments stream process. Seller of this claim can actively trade
in the market according to strategy $\varphi$. Since we do not
restricted to self-financing strategies this strategy can generate during period $(0,T]$ some cost which is defined below
\begin{defn}
The \emph{cost process} of an $L^2$-strategy $\varphi=(\phi,\eta)$
associated with a dividend process $D$ is given by
\begin{equation}\label{eq:cost-D}
C^D_t (\varphi) := \int_0^t \frac{1}{B_u} d D_u + V^*_t(\varphi)
- \int_0^t \phi^\top_t d S^*_t
\end{equation}
for $0 \leq  t \leq T$, provided that $\int_{]0,\cdot]}
\frac{1}{B_u} d D_u$ is square integrable. If the process $C^D$ is a constant,
then we say that $\varphi$ finances dividend process  $D$ (or
$\varphi$ is a self-financing for $D$). If $C^D$ is a martingale, then
we say that $\varphi$ is a mean-self-financing for $D$.
\end{defn}

The cost is the sum of three components, the first describes the
cumulated discounted dividends, the second is equal to the discounted wealth of the portfolio
and the third is  equal to the minus one multiply by   the discounted gains
from trading using $\varphi$. Obviously, the result of this simple
arithmetics deserves, from the
hedger point of view, for the name accumulated cost. In this paper we take a point of view of the hedger, i.e.,
a person who is obligated to deliver payments according to a
dividend process $D$. Therefore the gains from trading strategy
are subtracted because a  negative cost is a hedger's income.
Therefore  $V_t(\phi)$ could be interpreted as the wealth of portfolio
after delivering  payments described by  $D$ at/up to time $t$. Since $D$
matures at $T$,  it is quite natural to restrict our considerations
to the $0$--achieving strategies since the hedger of
$D$ should stop trading  after $T$. For more detailed motivation of above
definition we refer to M{\o}ller
\cite{mol2001}. We stress   that we don't assume that $D$ is positive
or has positive jumps. With a strategy $\varphi$ we connect another
quantity

\begin{defn}
The \emph{risk process} $R^D(\varphi)$ of a strategy $\varphi$
associated with $D$ is defined by
\[
    R^D_t(\varphi) = \mathbb{E} \left( ( C^D_T(\varphi) - C^D_t(\varphi))^2 | \mathcal{F}_t \right)
\]
for $0 \leq t \leq T$.
\end{defn}

F\"{o}llmer and Sonderman  \cite{folson1986} proposed the following
    definition of risk-minimality:
\begin{defn}
We say that $\varphi = (\phi,\eta)$ is a \emph{risk-minimizing
strategy} (for a dividend process $D$) if for any $t \in [\![0,T]\!]$
and any strategy $\widehat{\varphi} =
(\widehat{\phi},\widehat{\eta})$ satisfying
\begin{align}
\label{eq:portfolio-terminal}
V^*_T ( \varphi ) =  V^*_T ( \widehat{\varphi} ) \ \ \ \ \ \  \mathbb{P} - a.s., \\
\label{eq:portfolio-inter} \phi_s = \widehat{\phi}_s  \text{ for } s
\leq t, \text{ and  } \eta_s = \widehat{\eta}_s  \text{ for }  s
\leq t,
\end{align}
we have $R^D_t ( \widehat{\varphi}) \geq R^D_t (\varphi)$.
\end{defn}
Sometimes we call $\varphi$ a $D$-risk-minimizing strategy. A
strategy $\widehat{\varphi}$ satisfying
\eqref{eq:portfolio-terminal} and \eqref{eq:portfolio-inter} is
called the admissible continuation of $\varphi$ at time $t$.

F\"{o}llmer and Sonderman \cite{folson1986}  noticed that the
problem of finding of risk-minimizing hedging strategy can be
solved by using the Galtchouk-Kunita-Watanabe decomposition (GKW decomposition). However,
in \cite{folson1986} only the case of a single payment at maturity
is considered. The general case of payment streams was considered
by M{\o}ller \cite{mol2001} who noticed that
the GKW decomposition of martingale defined by
\begin{equation}\label{eq:intr-risk}
V^*_t := \mathbb{E} \left( \int_0^T \frac{1}{B_u} d D_u |
\mathcal{F}_t \right), \quad 0 \leq t \leq T,
\end{equation}
is a very useful in deriving $D$-risk-minimizing hedging
strategies. We recall that if $S^*$ and $V^*$ are square
integrable martingales, then $V^*$ can be uniquely decomposed
(Galtchouk-Kunita-Watanabe) as
\begin{align}\label{eq:GKW}
V^*_t = V^*_0 + \int_0^t \phi^D_u d S^*_u + L^D_t,
\end{align}
where $L^D$ is a zero mean martingale orthogonal to $S^*$ (i.e.,
$S^* L^D$ is a martingale) and $\phi^D$ is a predictable process
satisfying integrability assumption \eqref{eq:S-integr-cond}.
 M{\o}ller \cite{mol2001} obtained formula for a  risk-minimizing
hedging strategy for a dividend process $D$ using the GKW decomposition of $V^*$:
\begin{thm}\cite{mol2001}\label{thm:risk-min-general}
There exist a unique 0-achieving risk-minimizing strategy $\varphi =
(\phi,\eta)$ for the dividend process $D$. It is of the following
form
\[
    (\phi_t, \eta_t ) = \left(\phi^D_t, V^*_t - \int_0^t \frac{1}{B_u} d D_u - \phi^D_t S^*_t
    \right),
\]
where $\phi^D$ is from GKW decomposition of $V^*$. The risk
process of risk-minimizing strategy $\varphi$ is given by $R^D_t(\varphi) =
\mathbb{E} \left(  (L^D_T - L^D_t)^2| \mathcal{F}_t \right)$.
\end{thm}
Subsequently,  we consider only  the  0-achieving risk-minimizing strategies for dividend processes.
To shorten notation, we sometimes omit description
 "0-achieving" in a risk-minimizing strategy.
\begin{rem}
F\"{o}lmer and Sonderman \cite{folson1986} restrict themselves to finding
risk-minimizing strategies such that $V^*_T ( \varphi ) = X$. This is a
consequence,  as
explained by M{\o}ller \cite{mol2001}, of considering only single payments at maturity and
slightly different definition of the cost of strategy which is
independent of the claim $X$.
Note that this modification yields that for a 0-achieving risk-minimizing strategy $\varphi$ we have
\[
	V_t( \varphi)
=
B_t
\mathbb{E}\left(
\int_t^T \frac{1}{B_u} d D_u
\big|
\mathcal{F}_t
\right),
\]
i.e., the wealth of the 0-achieving risk-minimizing strategy is equal to the ex-dividend price of the claim with dividend process $D$.
\end{rem}

\section{Risk-minimizing hedging strategy on a general market with ratings}\label{sec:lmm-g}


 In this section we apply results described in the previous section to solving risk-minimization problem on  the general market with ratings. We consider a market with time horizon $T^* < \infty$ and with
$d$ primary risky assets with the price process denoted by $S$ and a money
account with the price process $B$.

The dynamics of discounted price process  $S^*_t :=  S_t / B_t$ is given by
\begin{equation}\label{eq:sde-lmm}
d S^*_t = \sigma_t d W_t  +
\int_{\mathbf{R}^n}
\!\!\!F_t(x) \widetilde{\Pi} (dx,dt) + \sum_{i,j \in \mathcal{K}
,j\neq i} {\rho^{i,j}_t} d M^{i,j}_t, \quad S^*_0 = s.
\end{equation}
Here, $W$ is an $n$ dimensional Wiener process, by $\widetilde{\Pi}$ we denote a compensated random measure
 \[
    \widetilde{\Pi}(dx , du) := \Pi (dx , du) - Q(dx,du ),
 \]
where $\Pi(dx,dt)$ is assumed to be an integer-valued random measure
on $\mathcal{B}(\mathbf{R}^n) \otimes \mathcal{B}(\mathbf{R}_+)$ with the unique compensator $Q(dx ,dt)$ which has density $\nu = (\nu_u)$ with respect to time, i.e.,
 \[
    Q(dx,du ) = \nu_u(dx)du,
 \]
 where $\nu_u$ is a L\'{e}vy measure.
Processes $M^{i,j}$ are driven by an additional source of
uncertainty, i.e., by a c\`{a}dl\`{a}g process $C$ taking values in a
finite set $\mathcal{K} = \set{1, \ldots, K}$ which can be
interpreted as a credit rating of some corporate or as a state of
economy. For each $i \in \mathcal{K}$ we define the process
\begin{align*}
H^i_t := 1_\set{i} (C_{t})
\end{align*}
 which is a c\`{a}dl\`{a}g process,
taking values in $\set{0,1}$, indicating in which state the process
$C$ is at given time $t$. In such models $\sigma$ has often a form
\[
    \sigma_t = \sum_{i \in \mathcal{K} } H^i_{t-}  \widehat{\sigma}^i_t,
\]
where $\widehat{\sigma}^i$ depends only on $S$. The same remark concerns $F$.
Moreover, we define processes $H^{i,j}$
by letting
\[
H^{i,j}_t := \sum_{0 < u \leq t} H^i_{u-} H^j_u.
\]
This process counts the number of transition from state $i$ to $j$ during
time interval $[0,t]$. Therefore $H^{i,j}$ are mutually orthogonal point processes
with jumps equal to $1$.

In this paper we make the standing assumptions:

\textbf{Assumption EI.} There exist nonnegative bounded
processes $\lambda^{i,j}$, $i,j \in \mathcal{K}$, $j \neq i $,
such that processes $M^{i,j}$ defined by
\begin{equation}\label{eq:M-ij}
    M^{i,j}_t = H^{i,j}_t - \int_{0}^{t} \lambda^{i,j}_u du
\end{equation}
are martingales.

\textbf{Assumption NCJ.} We require that the
 random measure $\Pi $ and the processes $H^{i,j}$, for $i$, $j
\in \mathcal{K}$, $i \neq j$, have no common jumps, i.e., for
every $t > 0 $ and every $b>0$,
  \begin{equation}\label{eq:hip-A}
\int_0^t \int_{ \norm{ x} > b}
 \Delta H^{i,j}_u \Pi( dx, du) =0 \quad  \mathbb{P} - a.s.
\end{equation}

\textbf{Assumption INT.} The process $\sigma$ is
predictable with values in the space of matrices of dimension $ d
\times n $, $F$ is a mapping from $\Omega \times \mathbf{R}_{+} \times \mathbf{R}^n$ with values in $\mathbf{R}^d$ which is  $\mathcal{P} \otimes \mathcal{B}(\mathbf{R}^n)$ measurable,  and $\rho^{i,j}$ are predictable processes with values
in the space of vectors of dimension $d$ satisfying
\begin{equation}\label{eq:int-cond}
\mathbb{E} \!\left( \!\int_0^{T^*} \!\!\left( \norm{\sigma_t }^2
\!\!+\!\!\int_{\mathbf{R}^n} \!\!\norm{F_t (x) }^2 \nu_t(dx)  +
\!\!\!\! \sum_{i,j \in \mathcal{K}: j \neq i } \norm{\rho^{i,j}_t}^2
\!\lambda^{i,j}_t \right) \!dt \!\right) < \infty.
\end{equation}

One of consequences of Assumption EI is that $H^{i,j}$ are counting processes with rating dependent intensities $\lambda^{i,j}$. It is worth to note that
\[
\lambda^{i,j}_t=
H^{i}_{t-} \lambda^{i,j}_t,
\quad
dt \times d\mathbb{P} \ a.e. ,
\]
 so $\lambda^{i,j}_t$ depends on rating at time $t-$ but in general it may also depend on the trajectory of rating process $C$ on the interval $[\![0, t [\![$.
Assumption INT implies that $S^*$ is a square integrable
martingale, so $\mathbb{P}$ is a martingale measure. We denote by $G$ the matrix-valued stochastic
process
\begin{equation}\label{eq:pred-cov-i}
G_t := \left( \sigma_t  ( \sigma_t )^\top  +
\int_{\mathbf{R}^n} F_t(x) (F_t(x))^\top \nu_t(dx)   + \left(
\sum_{i,j \in \mathcal{K} : j \neq i}  {\rho^{i,j}_t} \left(
\rho^{i,j}_t \right)^\top \lambda^{i,j}_t\right) \right).
\end{equation}
The process $G$ exists and is
finite by Assumption INT. One could notice that $G$ represents the density with respect to time of
predictable covariations  of $S^*$, i.e., the matrix of predictable covariations  of $S^*$ under
$\mathbb{P}$ is given by
\[
d \langle S^* \rangle_t = G_{t} dt.
\]

The theorem below is the main result of this section and states that,
under appropriate conditions, the risk-minimizing strategy can be
calculated by solving a linear system of equations. We find an explicit form of risk-minimizing strategy in terms of coefficients which appear in formula describing dynamics of prices  and in terms of processes which appear in a martingale representation of discounted cumulated payments.
\begin{thm}\label{thm:risk-min-spec}
Fix a dividend process $D$. Assume that the square integrable random variable $X := \int_0^T
\frac{1}{B_u} d D_u $ representing discounted cumulated payments up
to maturity time $T$ has the representation
\begin{equation}\label{eq:repr-H-P}
X = \mathbb{E} (X) + \int_{0}^T \delta_t^\top d W_t +
\int_{0}^T \int_{\mathbf{R}^n} J_t(x) \widetilde{ \Pi }(dt, dx) +
\sum_{i,j \in \mathcal{K} : j \neq i} \int_0^T \gamma^{i,j}_t d M^{i,j}_t.
\end{equation}
Then
there exists a  0-achieving risk-minimization strategy $\varphi = (\phi,\eta)$ for the dividend process $D$. The component $\phi$ of this strategy is the predictable version of the solution to the
linear system
\begin{align}\label{eq:GKW-LSE}
    G_t \phi_t = A_t,
\end{align}
where  a predictable process $A$ is given by the formula
\begin{align}\label{eq:Y-def}
A_t :=  \sigma_t  \delta_t + \int_{\mathbf{R}^n} F_t  (x)
J_t(x) \nu_t(dx) +
  \sum_{i,j \in \mathcal{K} : j \neq i} {\rho^{i,j}_t}  \gamma^{i,j}_t  \lambda^{i,j}_t.
\end{align}
The second component $\eta$ of strategy $\varphi$ is given by
\[
\eta_t = V^*_t - \int_0^t \frac{1}{B_u} d D_u - \phi_t^\top S^*_t.
\]
Moreover, the dynamics of $L^X$ has the form
\begin{align}
\label{eq:LX}
d L^X_{t} &= \left( \delta_t^\top - \phi_t^\top  \sigma_t  \right) d W_t  +
\int_{\mathbf{R}^n} \left( J_t(x) - \phi_t^\top  F_t (x) \right) \widetilde{ \Pi} (dt,dx)
\\ \nonumber
& + \sum_{i,j \in \mathcal{K} : j \neq i}  \left( \gamma^{i,j}_t -  \phi_t^\top
{\rho^{i,j}_t} \right)dM^{i,j}_t.
\end{align}
\end{thm}
\begin{proof}[Proof of Theorem \ref{thm:risk-min-spec}]
 By Theorem \ref{thm:risk-min-general} we only need to find  the explicit form of  GKW decomposition of $V^*$ defined by  \eqref{eq:intr-risk}.
  The GKW decomposition of $V^*$ always exists since $V^*$ is a square integrable martingale.
 Thus, we are looking for $\phi$ and  $L^X$
such that
\begin{equation}\label{eq:gkw-dec}
d  V^*_t = \phi_t^\top d S^*_t + d L^X_t,
\end{equation}
where $L^X$ is a martingale orthogonal to the martingale $S^*$, that is
$\langle L^X, S^{*(k)} \rangle= 0$ for each $k =1 , \ldots d $.
Suppose that $\phi$ satisfies \eqref{eq:gkw-dec},
then using representations \eqref{eq:repr-H-P} and  \eqref{eq:sde-lmm}
we have
\begin{align*}
 d L^X_t &= \delta_t^\top d W_t  + \int_{\mathbf{R}^n} J_t(x)
\widetilde{ \Pi } (dt, dx) + \sum_{i,j \in \mathcal{K} : j \neq i}  \gamma^{i,j}_t d
M^{i,j}_t
\\
& - \phi_t^\top \left(
\left(\sigma_t d W_t + \int_{\mathbf{R}^n} F_t(x) \widetilde{
\Pi } (dt,dx)  \right) + \sum_{i,j \in \mathcal{K} : j \neq i}
{\rho^{i,j}_t} d M^{i,j}_t \right).
\end{align*}
Rearranging yields
\begin{align*}
d L^X_t &= \left( \delta_t^\top - \phi_t^\top  \sigma_t  \right) d W_t  +
\int_{\mathbf{R}^n} \left( J_t(x) - \phi_t^\top  F_t  (x) \right) \widetilde{ \Pi } (dt,dx)
\\
& + \sum_{i,j \in \mathcal{K} : j \neq i}  \left( \gamma^{i,j}_t -  \phi_t^\top
{\rho^{i,j}_t} \right)dM^{i,j}_t.
\end{align*}
Now we want to find $\phi$  such that $L^X$ is orthogonal to each
$S^{*(k)}$.
First notice that the angle bracket of the continuous martingale
parts of $S^{* c (k)} $ and $L^{X,c}  $ has the form
\[
d \langle  S^{* c (k)} ,  L^{X,c} \rangle_t = \left(  \sigma_t  \right)_{\!k} \left(
\delta_t^\top - \phi_t^\top
\sigma_t  \right)^{\!\!\!\top} dt,
\]
where for a matrix $E$ by $(E)_k$ we denote the operation of taking
$k$--th row from the matrix $E$. By assumption NCJ
\begin{align*}
\Delta S^{*(k)}_t  \Delta L^X_t &=
\int_{\mathbf{R}^n}\left( F_t (x) \right)_{\!k} \left( J_t(x) - \phi_t^\top   F_t( x) \right)^{\!\!\top} \Pi (dt,dx)
\\
& + \sum_{i,j \in \mathcal{K} : j \neq i} ( {\rho^{i,j}_t})_k   \left( \gamma^{i,j}_t
- \phi_t^\top   {\rho^{i,j}_t} \right)^{\!\!\!\top} \Delta H^{i,j}_t.
\end{align*}
Therefore, since $X \in L^2(\mathbb{P})$, the following condition is equivalent to the strong
orthogonality 
\begin{align*}
\mathbf{0}&= \left(  \sigma_t
\right) \left( \delta_t^\top - \phi_t^\top    \sigma_t \right)^{\!\!\top}
\\
& + \int_{\mathbf{R}^d} \left(
F_t(x) \right)\left( J_t(x) - \phi_t^\top   F_t(x) \right)^{\!\!\top} \nu_t(dx) +
\sum_{i,j \in \mathcal{K} : j \neq i}      {\rho^{i,j}_t} \left(
\gamma^{i,j}_t - \phi_t^\top     {\rho^{i,j}_t}
\right)^{\!\!\!\top} \lambda^{i,j}_t
\end{align*} where $\mathbf{0} =
(0, \ldots, 0)^\top$. This equality 
%
simplifies to
\begin{align}
\nonumber &&   \left( \sigma_t
(\sigma_t)^{\!\top} + \int_{\mathbf{R}^d} F_t (x) (F_t (x))^{\!\top}
\nu_t(dx) + \left( \sum_{i,j \in \mathcal{K} : j \neq i}
{\rho^{i,j}_t} \left( \rho^{i,j}_t \right)^{\!\!\top}   \lambda^{i,j}_t
\right) \right) \phi_t
\\
\label{eq:GKW-LSE-2}
&& =  \left( \sigma_t  \delta_t
+ \int_{\mathbf{R}^n} F_t  (x)J_t(x) \nu_t(dx) +
  \sum_{i,j \in \mathcal{K} : j \neq i} {\rho^{i,j}_t}  \gamma^{i,j}_t  \lambda^{i,j}_t
\right).
\end{align}
Using  definition of $G$ and $A$ (i.e., \eqref{eq:pred-cov-i} and \eqref{eq:Y-def})
we see that  \eqref{eq:GKW-LSE-2} is, in fact, \eqref{eq:GKW-LSE}.
Therefore, the existence of GKW decomposition implies that $A_t \in \text{Im} G_t$. Taking $\phi$ as a minimum norm solution to \eqref{eq:GKW-LSE} we have an explicit form of $\phi$ in GKW decomposition. Now, the assertion of theorem follows from Theorem \ref{thm:risk-min-general}.
\end{proof}
\begin{cor}\label{cor:D-rms}
Assume that $G_t$ is invertible for every $t \in [\![0,T]\!]$.
Then the
component $\phi$ of 0-achieving $D$-risk-minimizing strategy is given by the following
formula
\begin{equation}\label{eq:risk-min-strat}
\phi_t =  \!\left( G_t
\right)^{\!-\!1} \!\!\left(\! \sigma_t   \delta_t \!+\!
\int_{\mathbf{R}^n} \!\!\! F_t  (x) J_t(x) \nu_t(dx) +
  \!\!\!\sum_{i,j \in \mathcal{K} : j \neq i} \!\!\!{\rho^{i,j}_t}  \gamma^{i,j}_t  \lambda^{i,j}_t \!
\right)\!\!.
\end{equation}
\end{cor}
\begin{proof}
It is an immediate consequence of invertibility of $G$ and \eqref{eq:GKW-LSE}.
\end{proof}
\noindent Theorem  \ref{thm:risk-min-general} and \eqref{eq:LX} imply
\begin{cor}
The risk process of $D$-risk-minimizing strategy has the form
\begin{align*}
R^D_{t} &= \mathbb{E}\bigg( \int_t^T \norm{ \delta_u^\top - \phi_u^\top  \sigma_u  }^2 d u  +
\int_t^T \int_{\mathbf{R}^n} \left( J_u(x) - \phi_u^\top  F_u (x) \right)^2 \nu_u(dx) du
\\ \nonumber
& + \sum_{i,j \in \mathcal{K} : j \neq i}  \int_t^T  \left( \gamma^{i,j}_u -  \phi_u^\top
{\rho^{i,j}_u} \right)^2 \lambda^{i,j}_u du \bigg| \mathcal{F}_t \bigg) .
\end{align*}
\end{cor}
We can also formulate conditions ensuring attainability of dividend processes. For definition of $\mathbb{P}$-admissibility and attainability of dividend process $D$ we refer to \cite[Def. 16.13 and 16.9]{jaknie2011}. In \cite[Prop. 16.14]{jaknie2011} it is proved that $\varphi$ is $\mathbb{P}$-admissible and replicates $D$ if and only if $V_t(\varphi) = B_t \mathbb{E} \left( \int_t^T\frac{1}{B_u} d D_u| \mathcal{F}_t\right)$.
\begin{thm}\label{prop:attain}
Let $D$ be a dividend process with representation \eqref{eq:repr-H-P}.
There exists a $\mathbb{P}$-admissible strategy $\varphi=(\phi,\eta)$ which replicates $D$ if and only if there exists a predictable process $\phi$ satisfying:
\begin{align*}
	 (\sigma_t)^\top \phi_t &= \delta_t \quad \quad \quad dt \times d\mathbb{P} \ a.e., \\
   (F_t (x))^\top \phi_t &=  J_t(x) \quad dt \times \nu_t(dx) \times d\mathbb{P} \ a.e., \\
(\rho^{i,j}_t)^\top \phi_t &= \gamma^{i,j}_t, \quad \quad \lambda^{i,j}_t dt \times d\mathbb{P} \ a.e. \quad \forall i,j \in \mathcal{K}, i\neq j.
\end{align*}
\end{thm}
\begin{proof}
    We prove in Lemma \ref{lem:replication} in Appendix that the existence of  $\mathbb{P}$-admissible strategy $\varphi=(\phi,\eta)$ which replicates $D$ is equivalent to the fact that the process $L^X$ is equal to zero. This fact and \eqref{eq:LX} give the assertion.
\end{proof}
If the compensator of $\Pi$ is a random measure such that all $\nu_t$ has the same finite support (e.g., if $\Pi$ is a Poisson  random measure with L\'{e}vy measure $\nu$ with finite support), then
 we can give sufficient conditions for attainability of an arbitrary square integrable dividend process $D$. To formulate these conditions we introduce
convenient notation: for vectors $a_1, \ldots, a_q$ by $[a_i]_{i=1}^q$ we denote the matrix created from theses vectors by setting them in columns, i.e.,
\[
	[a_i]_{i=1}^q
=[a_1 a_2 \ldots a_q].
\]
\begin{prop}
    Assume that  for every $t \in [\![0,T]\!]$ the measure $\nu_t$ has the same finite support, i.e.,
    \begin{align}\label{eq:support-cond}
        \text{supp } \nu_t = \set{x_1, \ldots, x_q},
        \quad q\in \mathds{N} .
    \end{align}
     Let
    \begin{align*}
    L^i_t :=\left[
     \sigma_t, \left[F_t(x_k) \right]_{k=1}^q, \left[\rho^{i,j}_t \right]_{ j \in \mathcal{K}, j\neq i}
    \right].
    \end{align*}
    If,  for every $i\in\mathcal{K}$,
     \begin{equation}\label{eq:rank-cond}
      \text{rank }(L^i_t) = n+q+K-1  \quad \text{for every } t \in [0,T]
      \quad
      \mathbb{P} \ a.s. ,
     \end{equation}
     then every square integrable dividend process $D$ is attainable.
    Necessary condition for \eqref{eq:rank-cond} is  that number of assets is sufficiently large, i.e. $d \geq n+q+K-1$.
\end{prop}
\begin{proof}
Under the above assumptions, conditions formulated in Theorem \ref{prop:attain} transform, for every $i \in \mathcal{K}$, to the following system of finite number of equations:
\begin{align*}
(\sigma_t)^\top \phi_t &= \delta_t \quad \quad \quad dt \times d\mathbb{P} \ a.e., \\
   (F_t (x_k))^\top \phi_t &=  J(t,x_k) \quad dt \times d\mathbb{P} \ a.e. \text{ for} \quad k=1,\ldots, q, \\
(\rho^{i,j}_t)^\top \phi_t &= \gamma^{i,j}_t, \quad \quad \lambda^{i,j}_t dt \times d\mathbb{P} \ a.e. \quad \forall j \in \mathcal{K}, j\neq i,
\end{align*}
which can be written in the form
\[
    (L^i_t)^\top \phi_t = \left[
                            \delta^\top_t,
                              (J^\top(t,x_k))_{k=1}^q, ((\gamma^{i,j}_t)^\top)_{j \in \mathcal{K}, j\neq i}
                          \right]^\top.
\]
Since \eqref{eq:rank-cond} holds $\mathbb{P}$ a.s., the minimum norm solution   $\phi$ to the above system of linear equations exists and is a predictable process. Hence $\varphi = (\phi, \eta ) $ is an admissible portfolio which replicates $D$.
\end{proof}
\begin{rem}\label{rem:hedging}
a) It is also worth to mention that Theorem \ref{prop:attain} underlines the necessity of using  other
methods of hedging risk for non-attainable  dividend streams.
Thus, we consider  in this paper the risk-minimizing strategies.

b) Theorem \ref{prop:attain} also implies that condition \eqref{eq:support-cond} is to some extent  necessary for attainability of an arbitrary square integrable dividend process $D$.

c) As we have seen, the replication of arbitrary dividend process requires sufficiently large number of assets which are in some sense linearly independent. In many cases this might be considered as a very unrealistic situation. Thus, a hedger has to choose the risk sources he/she wants to hedge against. For example, if a hedger will choose to hedge against Brownian-risk (continuous-risk), then he/she will be looking
for $\phi_t$ such that only the first condition is satisfied, i.e.,
\[
	(\sigma_t)^\top \phi_t = \delta_t.
\]
Usually this corresponds to the classical delta-hedging.
Alternatively, a hedger might choose to hedge against credit-risk and then he/she will be interested in solving the following system of linear equations
\[
\left[\rho^{i,j}_t \right]^\top_{ j \in \mathcal{K}, j\neq i} \phi_t = \left[\gamma^{i,j}_t \right]^\top_{j \in \mathcal{K}, j\neq i}.
\]
This kind of strategy is called \emph{delta-hedging of credit risk}.
Note that corresponding rank conditions for solving above systems is less demanding than \eqref{eq:rank-cond}.

\end{rem}
In Theorem \ref{thm:risk-min-spec} we assume that the square integrable random variable $X = \int_0^T
\frac{1}{B_u} d D_u $ representing discounted cumulated payments up
to maturity time has certain martingale representation.
This might be considered as a very restricted assumption, however
this representation is ensured by a weak property of predictable representation of the triple ($W,$ $\widetilde{\Pi}$, $M$) (see He, Wang and Yan \cite[Def. 13.13]{hewangyan}).
\begin{defn}
We say that ($W,$ $\widetilde{\Pi}$, $M$) has a weak property of predictable representation with respect to $(\mathbb{F},\mathbb{P})$ if
every  square integrable $(\mathbb{F},\mathbb{P})$-martingale $N$ has the representation
\begin{align}\label{eq:mart-rep}
    N_t  = N_0 + \int_0^t \phi_u d W_u + \int_0^t  \int_{\mathbf{R}^n} \psi_u(x) \widetilde{\Pi}(du,dx)
    + \int_{0}^t \sum_{i,j \in \mathcal{K} : j \neq i} \xi_{u}^{i,j} d M^{i,j}_u,
\end{align}
where  $\phi_u$, $\psi_u(x)$, $\xi_{u}^{i,j}$ are predictable processes such that
\[
\mathbb{E}\left( \int_0^T |\phi_u |^2du\right) <\infty,
\quad
\mathbb{E}\left( \int_0^T  \int_{\mathbf{R}^n} |\psi_u(x)|^2  \nu_u(dx) du\right) <\infty,
\]
\[
\mathbb{E}\left( \int_0^T  | \xi_{u}^{i,j} |^2 \lambda^{i,j}_u du \right) <\infty
\quad i,j \in \mathcal{K}, i \neq j.
\]
\end{defn}
  This property   holds in many models, here we give two examples.

\begin{example}
Assume that $C$ is a time homogenous Markov chain with an intensity matrix $\Lambda$ which is independent of $W$ and $\widetilde{\Pi}$ is an independent Poisson random measure. Then a weak property of  predictable representation holds for ($W$,$\widetilde{\Pi}$,$M$) with respect to $(\mathbb{F},\mathbb{P})$, where $\mathbb{F} = \mathbb{F}^W \vee \mathbb{F}^\Pi \vee \mathbb{F}^C$.
This follows by analogous arguments as in
Becherer \cite{bech2006}.
For the convenience of the reader we repeat these arguments.
First, note that each of processes $W$, $\widetilde{\Pi}$, $M$ has the weak property of  predictable representation with respect to their own filtration and therefore by independence and strong orthogonality of $W$, $\widetilde{\Pi}$, $M$ we have martingale representation \eqref{eq:mart-rep} for martingales $N$ such that  $N_T = \I_E \I_F\I_G$ where $E \in \mathcal{F}^W_T$, $F \in \mathcal{F}^\Pi_T$, $G \in \mathcal{F}^C_T$. Finally, by standard approximation techniques, we obtain the asserted weak property of a  predictable representation of ($W$, $\widetilde{\Pi}$, $M$) with respect to $(\mathbb{F}^W \vee \mathbb{F}^\Pi \vee \mathbb{F}^C,\mathbb{P})$.
\end{example}
%
\begin{example}
Assume that on a filtered probability space  $(\Omega, \mathcal{F},\mathbb{F},\widehat{\mathbb{P}})$ we have given
 a Brownian motion $\widehat{W}$,
 a compensated integer-valued random measure $\widehat{\Pi}$ with $\widehat{\mathbb{P}}$-compensator denoted by $\widehat{\nu}_t(dx) dt$, and counting processes $H^{i,j}$ with $\widehat{\mathbb{P}}$-intensity processes denoted by $\widehat{\lambda}^{i,j}$, so
 \[
     \widehat{M}^{i,j}_t := H^{i,j}_t - \int_0^t \widehat{\lambda}^{i,j}_u du
 \]
 is a $\widehat{\mathbb{P}}$-martingale.
Assume that ($\widehat{W}$, $\widehat{\Pi}$, $\widehat{M}$)
has the weak property of predictable representation with respect to $(\mathbb{F},\widehat{\mathbb{P}})$
and that $\mathbb{P}$ is absolutely continuous with respect to $\widehat{\mathbb{P}}$ with the density process
\begin{align*}
d \eta_t &= \eta_{t-} \left( \beta_t d \widehat{W}_t + \int_{\mathbf{R}^n} (Y(t,x) - 1 )\widehat{\Pi}(dt,dx)  + \sum_{i,j \neq i} (\kappa^{i,j}_t-1) d \widehat{M}^{i,j}_t  \right), \\
\eta_0 &=1.
\end{align*}
Then ($W$,$\widetilde{\Pi}$,$M$)  has the weak property of predictable representation with respect to $(\mathbb{F},\mathbb{P})$
(see Jacod and Shiryaev \cite[Thm. III.5.24]{js1987} ),
where
\begin{align*}
W_t &:= \widehat{W}_t - \int_0^t \beta_u du,
\quad
\widetilde{\Pi}(dt,dx)  := \Pi(dt,dx)- Y(t,x)\widehat{\nu}_t(dx)dt,\\
{M}^{i,j}_t &:= H^{i,j}_t - \int_0^t  \kappa^{i,j}_u\widehat{\lambda}^{i,j}_u du, \quad i,j \in \mathcal{K}, i \neq j.
\end{align*}
\end{example}

Example 2 shows, in principle, that we can construct a model having weak property of predictable representation with mutual dependence between
$W,\Pi$ and $C$ starting from an independent ones.

\subsection{Semimartingale dividend process}
One of most the important class of dividend processes constitute processes $D$ which are semimartingales with the following decomposition
\begin{align}\label{eq:D-sem}
D_t = \xi^D \I_{t \geq  T} &+ \int_0^t g^D_u du + \int_{0}^t (\delta_u^D)^\top d W_u +
\int_{0}^t \int_{\mathbf{R}^n} J^D_u(x) \widetilde{ \Pi }(du, dx)
\\
\nonumber
& +
\sum_{i,j \in \mathcal{K} : j \neq i} \int_0^t \gamma^{D,i,j}_u d M^{i,j}_t
\end{align}
with appropriately integrable quintuple of processes  $(\xi^D, g^D, \delta^d, J^D, \gamma^D)$.  In this case finding risk minimization strategy boils down to finding martingale representation of
\begin{align}\label{eq:hatX}
\widehat{X}:= \frac{\xi^D}{B_T} + \int_0^T \frac{g^D_u}{B_u} du
\end{align}
instead of $X = \int_0^T
\frac{1}{B_u} d D_u $. Indeed, applying the previous results for $D$ yield
\begin{thm}\label{thm:risk-min-spec-2}
Fix a dividend process $D$ of the form \eqref{eq:D-sem},
where $(\xi^D, g^D, \delta^d, J^D, \gamma^D)$ satisfies
\[
\mathbb{E}\left( |\xi^D B^{-1}_T|^2 + \int_0^T | B^{-1}_u ( g^D_u + \delta^D_u )  |^2 du\right) <\infty,
\quad
\mathbb{E}\left( \int_0^T  \int_{\mathbf{R}^n} |B^{-1}_u J^D_u(x)|^2  \nu_u(dx) du\right) <\infty,
\]
\[
\mathbb{E}\left( \int_0^T  | B^{-1}_u \gamma_{u}^{D,i,j} |^2 \lambda^{i,j}_u du \right) <\infty
\qquad \rm{for} \quad i,j \in \mathcal{K}, i \neq j.
\]
Assume that the square integrable random variable $\widehat{X} := \frac{\xi^D}{B_T} + \int_0^T \frac{g^D_u}{B_u} du$ has the representation
\begin{equation}\label{eq:repr-hatX}
\widehat{X} = \mathbb{E} (\widehat{X}) + \int_{0}^T \widehat{\delta}_t^{\, \top} d W_t +
\int_{0}^T \int_{\mathbf{R}^n} \widehat{J}_t(x) \widetilde{ \Pi }(dt, dx) +
\sum_{i,j \in \mathcal{K} : j \neq i} \int_0^T \widehat{\gamma}^{\, i,j}_t d M^{i,j}_t.
\end{equation}
Then
there exists a  0-achieving risk-minimization strategy $\varphi = (\phi,\eta)$ for the dividend process $D$. The component $\phi$ of this strategy is the predictable version of the solution to the
linear system
\begin{align*}
    G_t \phi_t = \widehat{A}_t,
\end{align*}
where  a predictable process $\widehat{A}$ is given by the formula
\begin{align*}
\widehat{A}_t &:=   \sigma_t  (\widehat{\delta}_t + B_t^{-1} \delta^D_t) + \int_{\mathbf{R}^n} F_t  (x)
(\widehat{J}_t(x)  + B_t^{-1} J^D_t(x))\nu_t(dx)
\\
\nonumber
& \quad \    +
  \sum_{i,j \in \mathcal{K} : j \neq i} {\rho^{i,j}_t}  (\widehat{\gamma}^{\,i,j}_t + B_t^{-1}\gamma^{D,i,j}_t) \lambda^{i,j}_t.
\end{align*}
The second component $\eta$ of strategy $\varphi$ is given by
\[
\eta_t = V^*_t - \int_0^t \frac{1}{B_u} d D_u - \phi_t^\top S^*_t.
\]
Moreover, the dynamics of $L^X$ has the form
\begin{align*}
d L^X_{t} &= \left( (\widehat{\delta}_t + B_t^{-1}\delta^D_t )^\top - \phi_t^\top  \sigma_t  \right) d W_t  +
\int_{\mathbf{R}^n} \left( \widehat{J}_t(x) + B_t^{-1}J^D_t(x)- \phi_t^\top  F_t (x) \right) \widetilde{ \Pi} (dt,dx)
\\ \nonumber
& \quad + \sum_{i,j \in \mathcal{K} : j \neq i}  \left( \widehat{\gamma}^{\,i,j}_t + B_t^{-1} \gamma^{D,i,j}_t -  \phi_t^\top
{\rho^{i,j}_t} \right)dM^{i,j}_t.
\end{align*}
\end{thm}
\begin{proof}
Straightforward calculations show that if $\widehat{X}$ has representation \eqref{eq:repr-hatX} then  $X = \int_0^T B^{-1}_u d D_u$   has the representation
\begin{align} \label{eq:repr-H-P-N}
X = & \mathbb{E} (X) + \int_{0}^T \!\!\bigg(\widehat{\delta}_t  + \frac{\delta_t^D}{B_t} \bigg)^{\!\!\top} d W_t +
\int_{0}^T \!\!\int_{\mathbf{R}^n} \bigg( \widehat{J}_t(x) + \frac{J^D_t(x)}{B_t} \bigg)\widetilde{ \Pi }(dt, dx)
\\ \nonumber
& +
\sum_{i,j \in \mathcal{K} : j \neq i} \int_0^T \!\!\bigg( \widehat{\gamma}^{\,i,j}_t + \frac{\gamma^{D,i,j}_t}{B_t}\bigg)d M^{i,j}_t.
\end{align}
Applying  Theorem \ref{thm:risk-min-spec} we obtain the assertion of theorem.
\end{proof}
Note that $\widehat{X}$ can be connected with the ex-dividend price of $D$ in the following way
\[
\pi_t(D) := B_t \mathbb{E} \left( \int_t^T \frac{1}{B_u} dD_u | \mathcal{F}_t \right)
=B_t \mathbb{E} \left( \frac{\xi^D}{B_T} + \int_t^T \frac{g^D_u}{B_u} du | \mathcal{F}_t\right)
=B_t \left(\mathbb{E} ( \widehat{X} | \mathcal{F}_t)- \int_0^t \frac{g^D_u}{B_u} du \right)
\]
The second conditional expectation resembles the one from Feynman-Kac type theorems. We will investigate this similarity in the next section where we will consider a very flexible Markovian framework. This allows us to derive explicit formulae for risk-minimizing
strategies by means of solutions to some PIDE's as we will see in Section 5.
Using \eqref{eq:repr-H-P-N} we see that the  previous results from this
section could be rewritten for the semimartingale dividend processes with
\begin{align*}
\delta_t &= \widehat{\delta}_t  +  B_t^{-1}\delta^D_t, \\
J_t(x) &=\widehat{J}_t(x) + B_t^{-1} J^D_t(x), \\
\gamma^{i,j}_t &=\widehat{\gamma}^{\,i,j}_t + B_t^{-1}\gamma^{D,i,j}_t.
\end{align*}
For instance, Theorem \ref{prop:attain} concerning attainability can be reformulated as follows
\begin{thm}\label{prop:attain-2}
Let $D$ be a dividend process given by  \eqref{eq:D-sem} for which $\widehat{X}$ has representation  \eqref{eq:repr-hatX}.
There exists a $\mathbb{P}$-admissible strategy $\varphi=(\phi,\eta)$ which replicates $D$ if and only if there exists a predictable process $\phi$ satisfying:
\begin{align*}
	 (\sigma_t)^\top \phi_t &= \widehat{\delta}_t  +  B_t^{-1}\delta^D_t\quad \quad \quad dt \times d\mathbb{P} \ a.e., \\
   (F_t (x))^\top \phi_t &=  \widehat{J}_t(x) + B_t^{-1} J^D_t(x)\quad dt \times \nu_t(dx) \times d\mathbb{P} \ a.e., \\
(\rho^{i,j}_t)^\top \phi_t &= \widehat{\gamma}^{\,i,j}_t + B_t^{-1}\gamma^{D,i,j}_t, \quad \quad \lambda^{i,j}_t dt \times d\mathbb{P} \ a.e. \quad \forall i,j \in \mathcal{K}, i\neq j.
\end{align*}
\end{thm}

\section{Risk-minimization in the Markovian market models with ratings}

 In this section we specify a  Markovian market model. 
 We assume that information available to the market participants is modeled by a multidimensional process $(Y,C)$ given as a solution to some SDE in $\mathbf{R}^{d+p}\times \mathcal{K}$.
 The first $d$ components of $Y$ are assumed to be a process $S$ of prices of tradable risky assets, and the remaining $p$ components,  denoted by $R$,  represent economic environment such as interest rates, inflation or stochastic volatility. Thus
 \[
 	S = (Y^i)_{i=1}^d, \quad
  R = (Y^i)_{i=d+1}^{d+p}.
 \]
 $C$ is a stochastic process which has finite state space $\mathcal{K} =
\set{1 , \ldots , K }$ and that could be interpreted as a credit rating of corporate, so $C_u$ represents credit rating of corporate at time $u \leq T$.
On the market there is a money account with the price process, denoted by
$B$, depending on on economic conditions of market, i.e. $R$ and $C$. So $B$ is given as a unique solution to
\begin{equation}  \label{eq:bank-acc}
d B_u = \mathfrak{r}(u,R_{u-},C_{u-}) B_u du,\quad B_0 = 1,
\end{equation}
where $\mathfrak{r}$ is a measurable, deterministic and bounded function.
Also the
 evolution
of price depends on economic conditions of market which
are described by rating system and additional non-tradable risk factor process $R$.
So, we assume that our model is described by SDE in which the credit rating of corporate has influence on asset prices $S$ and other non-tradable factor $R$ by
changing drift and volatility. Moreover, a change in credit rating
from $i$ to $j$ at time $u$ causes a jump of size ${\rho^{i,j}_Y(u,Y_{u-})}$.
Therefore, the evolution of $(S,R,C)$ is given as solution of the
following SDE in $\mathbf{R}^{d+p} \times \mathcal{K}$ :
\begin{align}
  \nonumber
  d \left(
     \begin{array}{c}
       S_u \\
       R_u \\
     \end{array}
   \right)
 &=   \left(
     \begin{array}{c}
       \mu_S(u,S_{u-}, R_{u-}, C_{u-})  \\
       \mu_R(u,R_{u-}, C_{u-})  \\
     \end{array}
   \right)
   du +
   \left(
     \begin{array}{c}
       \sigma_S( u, S_{u-},R_{u-}, C_{u-})   \\
       \sigma_R( u, R_{u-}, C_{u-})   \\
     \end{array}
   \right)
   d W_u \\
   \nonumber  & + \int_{\mathbf{R}^n}
   \left(
     \begin{array}{c}
    F_S(u,S_{u-},R_{u-},C_{u-}, x ) \\
    F_R(u,R_{u-},C_{u-}, x ) \\
     \end{array}
   \right)
\widetilde{\Pi}(dx , du)
\\
 \label{eq:SDE-SR}
 & \
  + \sum_{ \substack{i,j=1 \\ j\neq i} }^K
    \left(
     \begin{array}{c}
  \rho^{i,j}_S(u,S_{u-},R_{u-}) \\
  \rho^{i,j}_R(u,R_{u-}) \\
     \end{array}
   \right)
  \I_\set{i}(C_{u-}) (d N^{i,j}_u - \lambda^{i,j}(u, S_{u-}, R_{u-}) du),
& \\
 \nonumber
 d C_u &= \sum_{ \substack{i,j=1 \\ j\neq i} }^K (j - i) \I_{i} (C_{u-}) d
 N^{i,j}_u ,
 \\
 \nonumber
 S_0 & = s, \quad R_0 = r , \quad C_0 = c \in \mathcal{K},
\end{align}
where $W$ is a standard $r$-dimensional Wiener process, $\widetilde{\Pi}(dx,du)$ is a  compensated Poisson random measure on
$\mathbf{R}^n \times [\![0,T]\!]$ with the intensity measure
$\nu(dx) du $ where $\nu$ is a L\'{e}vy measure, and
$N^{i,j}$  are counting point processes with intensities determined by
$\lambda^{i,j}$, bounded continuous functions in $(u,y)$, i.e., for the fixed  $(i,j)$, the process
\begin{align}\label{eq:mart-Mij-t}
\widetilde{M}^{i,j}_u := N^{i,j}_u - \int_0^u \lambda^{i,j}(v, Y_{v-}) dv
\end{align}
is a martingale. Note that
the martingales defined by  \eqref{eq:M-ij}  take the following form
\begin{equation}\label{eq:M-ij-markov}
    M^{i,j}_u = H^{i,j}_u - \int_0^u H^{i}_{v-} \lambda^{i,j}(v, Y_{v-}) dv,
\end{equation}
since
\begin{align}\label{eq:Hij}
    H^{i,j}_u = \int_{0}^u \I_{i} (C_{v-}) d
 N^{i,j}_v.
\end{align}
Moreover, we require that the
Poisson random measure $\Pi $ and the processes $N^{i,j}$, $i,j
\in \mathcal{K}$, $i \neq j$, have no common jumps, i.e., for
every $t > 0 $ and every $b>0$,
  \begin{equation} \label{eq:hip-A-SDE}
\int_0^u \int_{ \norm{ x} > b}
 \Delta N^{i,j}_v \Pi( dx, dv) =0 \quad  \mathbb{P} - a.s.,
\end{equation}
and for all $(i_1,j_1) \neq ( i_2,j_2)$,
 \begin{equation}
\label{eq:hip-B}
    \Delta N^{i_1,j_1}_u\Delta N^{i_2,j_2}_u =  0 \quad  \mathbb{P} - a.s.
\end{equation}
We also impose the following condition
\[
	\mu_S(u,s,r,c) = s \mathfrak{r}(u,s,r,c),
\]
which implies that discounted prices of tradable assets $S$ are local martingales under the probability $\mathbb{P}$.
The presence of $\rho^{i,j}$ in the above SDE adds extra flexibility to model, so it is very important from the point of view of possible applications.
For example, it enables us to introduce into a model a dependence of intensity of jumps $C$ at time $t$ on behavior of trajectory of process $C$ up to time $t-$.  This is the case  of semi-Markov processes where $\lambda^{i,j}$ at time $t$ depends on time which the process $C$ spends in the current state after the  last jump.
The process (say $R^1$) corresponding to this semi-Markovian dependence can be introduced in our framework by setting
\[
	d R^1_t = dt - \sum_{i,j \in \mathcal{K}: i \neq j } R^1_{t-} \I_{i}(C_{t-}) d N^{i,j}_t.
\]
So if we allow $\lambda^{i,j}$ to be a (non-constant) function of $R^1$ we obtain a semi-Markov model (see Section 6). \\
For convenience and brevity we introduce the following notation:
\begin{align*}
&Y_u := (S_u, R_u), \\
&z := (s,r), \quad
\mu_Y(u,z,c) := \left(
                  \begin{array}{c}
                    \mu_S(u,s,r,c) \\
                    \mu_R(u,r,c) \\
                  \end{array}
                \right), \quad
\sigma_Y(u,z,c) := \left(
                  \begin{array}{c}
                    \sigma_S(u,s,r,c) \\
                    \sigma_R(u,r,c) \\
                  \end{array}
                \right),
\\
&F_Y(u,z,c,x) := \left(
                  \begin{array}{c}
                    F_S(u,s,r,c,x) \\
                    F_R(u,r,c,x) \\
                  \end{array}
                \right), \quad
\rho^{i,j}_Y(u,y)
=
\left(
     \begin{array}{c}
  \rho^{i,j}_S(u,s,r) \\
  \rho^{i,j}_R(u,r) \\
     \end{array}
\right), \\
& a_Y(u,z,c) := \left(
                  \begin{array}{cc}
                    a_{SS}(u,z,c) & a_{SR}(u,z,c)\\
                    a_{RS}(u,z,c) & a_{RR}(u,z,c)\\
                  \end{array}
                \right)=
                \left(
                  \begin{array}{cc}
                    \sigma_S \sigma_S^\top(u,z,c) & \sigma_S \sigma_R^\top(u,z,c)\\
                    \sigma_R \sigma_S^\top(u,z,c) & \sigma_R \sigma_R^\top(u,z,c)\\
                  \end{array}
                \right). \quad
\end{align*}
Using this notation we can write SDE
\eqref{eq:SDE-SR}
as the following multidimensional  SDE  in $\mathbf{R}^{d+p} \times \mathcal{K}$
\begin{align}\label{eq:SDE-gen}
 d Y_u &=   \mu_Y(u,Y_u, C_u) du + \sigma_Y( u, Y_{u-}, C_{u-}) d W_u + \int_{\mathbf{R}^n} F_Y(u,Y_{u-},C_{u-}, x )
\widetilde{\Pi}(dx , du)
\\
\nonumber
& \
  + \sum_{ \substack{i,j=1 \\ j\neq i} }^K \rho^{i,j}_Y(u,Y_{u-}) \I_\set{i}(C_{u-}) (d N^{i,j}_u - \lambda^{i,j}(u, Y_{u-}) du),
\\
\nonumber
 d C_u &= \sum_{ \substack{i,j=1 \\ j\neq i} }^K (j - i) \I_{i} (C_{u-}) d
 N^{i,j}_u ,
 \\
 \nonumber
 Y_0 & = y , \quad C_0 = c \in \mathcal{K},
\end{align}
with the coefficients being
measurable functions
 $\sigma_Y( \cdot, \cdot,
\cdot ) : [\![0,T]\!] \times \mathbf{R}^{d+p} \times \mathcal{K}
\rightarrow  \mathbf{R}^{(d+p)\times r}$, $F_Y(\cdot,\cdot, \cdot, \cdot ) :
[\![0,T]\!] \times \mathbf{R}^{d+p} \times \mathcal{K} \times
\mathbf{R}^{n} \rightarrow \mathbf{R}^{d+p} $ and $\rho^{i,j}_Y(\cdot,\cdot)
: [\![0,T]\!] \times \mathbf{R}^{d+p}  \rightarrow \mathbf{R}^{d+p}$.
This
SDE is a non standard one since a driving noise depends on the
solution itself (as in Jacod and Protter \cite{jp1978}, Becherer and Schweizer \cite{bs2005}). That is, the noise
$(N^{i,j})_{i,j \in \mathcal{K} : j\neq i}$ is not given apriori, it
is also constructed.
The solution is a quintuple $(Y,C, (N^{i,j})_{i,j \in \mathcal{K} : j\neq i},W, \Pi)$ together with a filtered probability space $(\Omega, \mathbb{F}, \mathbb{P})$. This corresponds exactly to the concept of weak solutions of SDE's (see Rogers and Williams \cite{RW2000}) and uniqueness in law is suitable concept of uniqueness for weak SDE's.
On the other hand the existence and weak uniqueness of  weak solutions to SDE \eqref{eq:SDE-gen}
(started at $0$)
is also required to model a market in the proper way. Indeed,
if there is no uniqueness in law, then there is an ambiguity on a law of prices, so  we do not have a market model specified correctly.
Furthermore to establish connection between SDE's and related PIDE's (and Markov property) we
are also interested  in existence and uniqueness of solutions to the SDE \eqref{eq:SDE-gen} started at every $t \in [\![0,T]\!]$ from every $(y,c) \in \mathbf{R}^{d+p} \times \mathcal{K}$ and defined on interval $[\![t,T]\!]$.
Note that these weak solutions
$(Y^{t,y,c},C^{t,y,c}, ((N^{i,j})^{t,y,c})_{i,j \in \mathcal{K} : j\neq i}, W^{t,y,c}, \Pi^{t,y,c})_{ u \in [\![t,T]\!]}$
can be constructed on possibly different filtered probability spaces $(\Omega^{t,y,c}, \mathbb{F}^{t,y,c}, \mathbb{P}^{t,y,c})$.

The expectation with respect to $\mathbb{P}^{t,y,c}$ will be denoted by  $\mathbb{E}^{t,y,c}$.
By $(H^{i,j})^{t,y,c}$ and  $(M^{i,j})^{t,y,c}$ we denote processes defined by
\begin{align}
\nonumber
    (H^{i,j})^{t,y,c}_u & := \int_{t}^u \I_{i} (C^{t,y,c}_{v-}) d
 (N^{i,j})^{t,y,c}_v,
 \\
 \label{eq:M-ijtyc}
    (M^{i,j})^{t,y,c}_u &:= (H^{i,j})^{t,y,c}_u - \int_t^u \I_{i} (C^{t,y,c}_{v-}) \lambda^{i,j}(v, Y^{t,y,c}_{v-}) dv.
\end{align}
Sometimes, for brevity of notation, we will write $(Y,C)$ to denote $(Y^{0,y,c},C^{0,y,c})$ - a solution to SDE \eqref{eq:SDE-gen} started at $0$ from $(y,c)$ and write $\mathbb{E}$ instead of $\mathbb{E}^{0,y,c}$.
 As in \cite{jaknie2012} we can
connect
 solutions to SDE \eqref{eq:SDE-gen} with corresponding martingale problem.
Let $\mathcal{I}$ be a class of functions defined by
\begin{align*}
\mathcal{I}=&\left\{
v: [\![0,T]\!] \times \mathbf{R}^{d+p} \times
\mathcal{K} \rightarrow \mathbf{R}: v \text{ is measurable, and } \forall s \in \mathbf{R}^d \ \forall i \in \mathcal{K}
\right.
\\
&
\quad
\left.
\int_{\mathbf{R}^n}\!\!\!\!\left| v^i(
u, z + F_Y(u,z,i,x)) - v^i(u, z ) - \nabla v^i(u,z) F_Y(u,z,i,x)
 \right|\nu(dx) < \infty \right\}.
\end{align*}
Let   $\mathcal{C}^{1,2}  =
\mathcal{C}^{1,2}( [\![0,T]\!] \times \mathbf{R}^{d+p}  \times \mathcal{K})$
be the  space of all measurable functions $v : [\![0,T]\!] \times
\mathbf{R}^{d+p} \times \mathcal{K} \rightarrow \mathbf{R} $ such that
$ v(\cdot, \cdot, k ) \in  \mathcal{C}^{1,2}([\![0,T]\!] \times \mathbf{R}^{d+p}
)$  for every $ k \in \mathcal{K} $, and let $\mathcal{C}^{1,2}_c$ be a set
of functions $f \in \mathcal{C}^{1,2}$ with compact support.
Define the family of operators $(\mathcal{A}_t)_{t \in [\![0,T]\!]}$ acting on
$\mathcal{C}^{1,2} \cap \mathcal{I}$ by setting
%
\begin{align}
\label{eq:gener-SDE}&\mathcal{ A }_t v( z,c) :=
  \nabla v(u,z,c)\mu_Y(u,z,c)
+ \frac{1}{2}
 \Tr \left( a_Y(u,z,c) \nabla^2 v(u,z,c) \right)
\\ \nonumber
 & + \int_{\mathbf{R}^n}\left( v( t,y  + F_Y(u,z,c,x) ,c)
- v( t,y,c) - \nabla v(u,z,c) F_Y(u,z,c,x)
\right)\nu(dx)
\\
\nonumber & + \sum_{c' \in \mathcal{K} \setminus c } \left( v( u, z  +
\rho^{c,c'}_Y(u,z),c') - v(u,z,c) - \nabla v(u,z,c)\rho^{c,c'}_Y(u,z) \right)\lambda^{c,c'}(u,z).
\end{align}
Here, by $Tr$ we denote the trace operator, $a_Y(u,z,c) :=  \sigma_Y(u,z,c)  (\sigma_Y(u,z,c))^{\!\top}$\!, and
$\nabla^2 v$ is the matrix of second derivatives of $v$ with respect to the
components of $y$. By $\nabla v$($\nabla_S v$, $\nabla_R v$ respectively) we denote the row vector of partial derivatives  of function $v$ with respect $y$ ($s$,$r$ respectively).

\noindent We make the following assumption on coefficients of \eqref{eq:SDE-gen}: \\
\textbf{Assumption LG.}
Functions $\mu_Y$,  $\sigma_Y$, $F_Y$, $\rho^{i,j}_Y$  satisfy the following linear growth condition
  \begin{align*}\tag{LG}
  	|\mu_Y(u,z,c)|^2
	+
  	|\sigma_Y(u,z,c)|^2
  +
  \int_{\mathbf{R}^n}|F_Y(u,z,c,x)|^2 \nu(dx)
  +
  \sum_{i,j \neq i}
  |\rho^{i,j}_Y(u,z)|^2
  \leq K(1 + |z|^2),
  \end{align*}
\noindent
Assuming (LG), among others, we can prove  very important facts (see \cite{jaknie2012}) used in the sequel, such as:
\begin{enumerate}
  \item The component $Y^{t,y,c}$
   of any weak solution to SDE \eqref{eq:SDE-gen}, started from $(y,c)$ at $t \in [\![0,T]\!]$, and  constructed on some filtered probability space  $(\Omega^{t,y,c}, \mathbb{F}^{t,y,c}, \mathbb{P}^{t,y,c})$
      satisfies \begin{align*}\label{eq:sup-S2r}
        \mathbb{E}^{t,y,c} \left(\sup_{u \in [\![t,T]\!]} |Y^{t,y,c}_u|^{2} \right)< \infty.
    \end{align*}
\item $\mathcal{C}^{1,2}_c \subset \mathcal{I} $.
  \item For $v \in \mathcal{C}^{1,2}_c$ the mapping $(u,z,c) \mapsto \mathcal{A}_u v(y,c) $
  has quadratic growth.
\end{enumerate}
By applying the It\^{o} lemma and exploiting properties (1), (2) and (3) we obtain that for every $v \in \mathcal{C}^{1,2}_c$ the process  $M^v$ defined for $u \in [\![t,T]\!]$ by
\[
	M^v_u := v(Y^{t,y,c}_u,C^{t,y,c}_u) - \int_{t}^u \mathcal{A}_u v(Y^{t,y,c}_u,C^{t,y,c}_u)du
\]
is a $(\mathbb{F}^{t,y,c}, \mathbb{P}^{t,y,c})$ martingale (see \cite[Prop. 2.5 and Thm. 5.1]{jaknie2012}).
 We have, by construction, that $\mathbb{P}^{t,y,c}( Y^{t,y,c}_t = y , C^{t,y,c}_t = c) =1$, and this means that $\mathbb{P}_{t,y,c}$, the law of $(Y^{t,y,c}_u,C^{t,y,c}_u)_{ u \in [\![t,T]\!]}$ considered as the law on an appropriate Skorochod space, solves the time dependent martingale problem for $(A_u)_{u \in [\![t,T]\!]}$ started at $t$ from $(y,c)$.

Of course, in general, assumption (LG) is to weak to guarantee existence and uniqueness of solutions to  SDE \eqref{eq:SDE-gen}, therefore, to
present results as general as possible, we assume

\textbf{Assumption EUWS} (existence of  unique weak solution).
For each $(y,c) \in \mathbf{R}^{d+p} \times \mathcal{K}$ and $t \in [\![0,T]\!]$, the SDE \eqref{eq:SDE-gen} on the interval $[\![t,T]\!]$ has a unique weak solution on some filtered probability space $(\Omega^{t,y,c}, \mathbb{F}^{t,y,c}, \mathbb{P}^{t,y,c})$.
\smallskip

For some sufficient conditions ensuring that EUWS is satisfied we refer to
Kurtz \cite{kurtz2010} and Jakubowski and Niew\k{e}g\l owski \cite{jaknie2012} where the existence and weak uniqueness of SDE's is related with the existence and well-possedness of corresponding martingale problem. In particular assumption EUWS yields that martingale problem for $(\mathcal{A}_u)_{u \in [\![t,T]\!]}$  is well possed, and therefore
 $\set{ \mathbb{P}^{t,y,c}:  (s,c) \in \mathbf{R}^d \times \mathcal{K}}$
is  a Markov family (see \cite[Prop. 4.5]{jaknie2012}). So, for a bounded functional $h$, we have
\begin{align}\label{eq:row_Ma}
    \mathbb{E}^{0,y,c} \left( h \left( (Y^{0,y,c}_u, C^{0,y,c}_u )_{ u \in [\![t,T]\!] }  \right) | \mathcal{F}^{0,y,c}_t \right)
    =
    \mathbb{E}^{t,y',c'} \left( h \left( (Y^{t,y',c'}_u, C^{t,y',c'}_u )_{ u \in [\![t,T]\!] }  \right) \right)\!\!\big|_{(y',c') =  (Y^{0,y,c}_t, C^{0,y,c}_t )}
\end{align}
where equality holds  $\mathbb{P}^{0,y,c} \ a.s.$

Now  we establish  a 0-achieving risk-minimizing strategy for a rating sensitive claim (see \cite{jaknie2011}), i.e., for a payment stream $D$ given by
\begin{align}\label{eq:D}
D_t := h(Y_T, C_T)1_\set{t \geq T } + \int_{0}^t g(u,Y_{u-}, C_{u-}
) du + \sum_{i,j\neq i } \int_{0}^t  \delta^{i,j}(u, Y_{u-})
d H^{i,j}_u,
\end{align}
$0 \leq t \leq T$, where $h$, $g$, $\delta^{i,j}$ are measurable real-valued function such that
\begin{align}\label{eq:D1}
\forall t \leq T \qquad    \mathbb{E} \!\!\left( \int_t^T \frac{1}{B_u} d D_u  \right) < \infty. \quad
\end{align}
We connect this  with a problem of finding the \emph{ex-dividend price function}, i.e. a
function $v$  such that $v(t,Y_t,
C_t) = V_t$,
with the process $V$ defined by
\begin{equation}\label{eq:proc-ex-div}
V_t :=  \!B_t\mathbb{E} \!\!\left( \frac{\!h(Y_T, C_T)}{B_T} +
\!\!\int_{t}^T \!\!\frac{g(u,Y_{u-}, C_{u-} )}{B_u} du +
\!\!\sum_{i,j\neq i } \int_{t}^T \!\!\!\frac{\delta^{i,j}(u,
Y_{u-}) }{B_u}d H^{i,j}_u \bigg| \mathcal{F}_t
\!\!\right)\!.
\end{equation}
Note that the existence of such function $v$ follows immediately from the Markov property of $(Y,C)$ with respect to $\mathbb{F}$.
Under assumption EUWS, by \eqref{eq:row_Ma}, we have
\[
	V_t = v(t,Y_t, C_t),
\]
where
\begin{align}\label{eq:v-def}
v(t,y,c)
:= \!\mathbb{E}^{t,y,c} \!\bigg(
& \frac{h(Y^{t,y,c}_T, C^{t,y,c}_T)}{B^{t,y,c}_{T}   } +
\!\!\int_{t}^T \!\! \frac{g(u,Y^{t,y,c}_{u-}, C^{t,y,c}_{u-} ) }{B^{t,y,c}_{u} }du
\\
\nonumber
&
+
\!\!\sum_{i,j\neq i } \int_{t}^T \!\!\! \frac{\delta^{i,j}(u,
Y^{t,y,c}_{u-}) }{B^{t,y,c}_{u} }d H^{i,j}_u
\!\!\bigg)\!,
\end{align}
    with $B^{t,y,c}_{u} $ defined by
\begin{align*}
B^{t,y,c}_{u} = \exp\bigg(\int_t^u \!\!\!\mathfrak{r}( v, R^{t,y,c}_{v}, C^{t,y,c}_{v}  )dv\bigg).
\end{align*}
So $B^{t,y,c}$ is the unique solution, on  the interval $[\![t,T]\!]$,  of ODE
\begin{align}  \label{eq:def-B1}
    d B^{t,y,c}_{u} &= \mathfrak{r}( u, R^{t,y,c}_{u}, C^{t,y,c}_{u}  ) B^{t,y,c}_{u} du,
    \quad
    B^{t,y,c}_{t} =1.
\end{align}
Let us notice that the law of $B^{t,y,c}_{u} $ is the same as the law of $\frac{B_u}{B_t}$  on the set $\set{Y^{0,y',c'}_t=y, C^{0,y',c'}_t=c}$  for arbitrary $(y',c')$.

In general, we do not know anything about regularity of $v$. The next theorem states that for a sufficiently regular function $v$ we can find a  0-achieving risk-minimizing strategy for $D$ written explicitly in terms of this function $v$.
\begin{thm}\label{thm:risk-min-markov}
Consider the dividend  process $D$ given by \eqref{eq:D}. Assume EUWS and that:

\noindent i)
There exists $m \geq 1$
such that
    \begin{align}\label{eq:sup-S2}
        \mathbb{E} \left(\sup_{t \in [\![0,T]\!]} |Y_t|^{2m} \right)< \infty,
    \end{align}
and moreover the following growth condition for this $m$ holds
\begin{align}\label{eq:LG-D}
|h(z,i)|^2 + |g(u,z,i)|^2 + \sum_{ j \in \mathcal{K}, j \neq i} | \delta^{i,j} (u,z)|^2 \leq K (1 + |z|^{2m}).
\end{align}
ii)
There exist an ex-dividend price function  $v$ which belongs to $ \mathcal{C}^{1,2} \cap \mathcal{I}$. 

\noindent Then
there exists a  0-achieving risk-minimization strategy $\varphi = (\phi,\eta)$ for the dividend process $D$. The component $\phi$ of this strategy is the predictable version of solution to the
linear system
\[
    G(t,Y_{t-},C_{t-}) \phi_t = A(t,Y_{t-},C_{t-}),
\]	
where
\begin{align*} 
G(u,z,i) := &
a_{SS}(u,z,i)
+
\int_{\mathbf{R}^{n}} \!\!\!(F_S F_S {^{\!\!\!\top}} )(u, z , i,x) \nu(dx)
 + \!\!\!\!\sum_{j \in \mathcal{K} : j \neq i}  \!\!\!(\rho^{i,j}_S\rho^{i,j}_S{^{\!\top}}) (u, z )  \lambda^{i,j}(u,z),
\\
A(u,z,i) := & a_{SS}(u,z,i)(\nabla_{\!S} v(u,z,i))\!^\top
+ a_{SR}(u,z,i) (\nabla_{\!R} v(u,z,i))\!^\top
\\
 & + \!\!\int_{\mathbf{R}^{n}} \!\!\!\!F_S(u,z,i,x)
\!\left(  v (u,z\! + \!F_Y(u,z,i,x),i
 ) - v(u, z ,i)
\right)
 \nu(dx) \\
 & +\!\!\!\!
  \sum_{j \in \mathcal{K} : j \neq i} \!\!\!\!\rho^{i,j}_S(u,z)
\!\left(v(u, z  \!+\! \rho^{i,j}_Y(u,z),j) \!-\! v(u, z ,i) +
\delta^{i,j}(u,z) \right)
   \lambda^{i,j}(u,z)
\end{align*}
and
\[
\eta_t = \frac{v(t,Y_t,C_t)1_\set{ t < T } - \phi_t^\top S_t}{B_t}.
\]
\end{thm}
\begin{proof}
Let $X = \int_0^T \frac{1}{B_u} d D_u $. We find the martingale representation of $X$ of the form  \eqref{eq:repr-H-P}, and then apply Theorem \ref{thm:risk-min-spec} to obtain the asserted formulae.
Using i), the Cauchy-Schwartz inequality and the isometry formula for stochastic integrals we see that  $X \in L^2$.
Thus, $M_t := \mathbb{E}(X|\mathcal{F}_t)$ is a martingale in $\mathcal{H}^2$ - the space of square integrable martingales.
By  definition of $V$ (i.e. \eqref{eq:proc-ex-div}), $M$ can be represented as
\begin{align*}
M_t &=  \frac{V_t}{ B_t} + \int_{0}^t \frac{g(u,Y_{u-}, C_{u-}
)}{B_u} du + \sum_{i,j\neq i } \int_{0}^t  \frac{\delta^{i,j}(u, Y_{u-})}{B_u}
d H^{i,j}_u \\
&=  \frac{v(t,Y_t,C_t)}{ B_t} + \int_{0}^t \frac{g(u,Y_{u-}, C_{u-}
) + \sum_{i,j\neq i }  \delta^{i,j}(u, Y_{u-})
H^i_{u-} \lambda^{i,j}_u
}{B_u} du \\
& \quad + \sum_{i,j\neq i } \int_{0}^t  \frac{\delta^{i,j}(u, Y_{u-})}{B_u}
d M^{i,j}_u.
\end{align*}
Using ii), the It\^o formula,  \eqref{eq:SDE-gen} and \eqref{eq:def-B1} we obtain
\begin{align*}
d\frac{v(t,Y_t,C_t)}{ B_t}
&= \frac{1}{ B_t} (dv(t,Y_t,C_t) - \mathfrak{r}(t,Y_{t-},C_{t-}) v(t,Y_t,C_t) dt)
\\
&=\frac{1}{ B_t} (dM^v_t + [(\partial_t + \mathcal{A}_t- \mathfrak{r})v](t,Y_{t-},C_{t-}) dt),
%
\end{align*}
where $M^v$ is given by
\begin{align}
\label{eq:def-Mv} &M^v_t := \int_0^t \sum_{i} H^i_{u- } \nabla
v(u, Y_{u-},i)\sigma^i_Y(u,Y_{u-}) d W_u
\\
\nonumber
& + \int_0^t \sum_{i}  H^i_{u- }
 \int_{\mathbf{R}^n} (v(u , Y_{u-} +  F_Y(u,Y_{u-}, i, x ),i) - v(u , Y_{u-},i)) \widetilde{\Pi} (du,dx)
\\
\nonumber & +  \int_0^t \sum_{i,j \neq i } (v(u , Y_{u-} +
{\rho^{i,j}_Y}(u,Y_{u-}),j) - v(u , Y_{u-},i) ) H^{i}_{u-}d M^{i,j}_u .
\end{align}
Hence we obtain
\begin{align}\label{eq:spec-sem}
M_t &= v(0,Y_0,C_0) + \int_0^t \frac{1}{ B_u} (dM^v_u + \sum_{i,j\neq i }  \delta^{i,j}(u, Y_{u-}) d M^{i,j}_u) + A_t,
\end{align}
where
\begin{align}\label{eq:A}
A_t=
\int_{0}^t \frac{1}{B_u}
\left(
\sum_{i} H^i_{u-}
\left[ [(\partial_t + \mathcal{A}_u-\mathfrak{r})v](u,Y_{u-},i)  + g(u,Y_{u-}, i
) + \sum_{j\neq i }  \delta^{i,j}(u, Y_{u-})
 \lambda^{i,j}(u,Y_{u-})
 \right]
 \right)
 du.
\end{align}
We see that RHS of \eqref{eq:spec-sem} is a special semimartingale with the unique predictable finite variation part given by $A$. On the other hand, this special semimartingale is equal to $M$  which is a martingale and therefore  $A$
is equal to zero (by uniqueness of canonical decomposition of special semimartingales). This implies that
\[
X=M_T = v(0,Y_0,C_0) + \int_0^T \frac{1}{ B_u} (dM^v_u + \sum_{i,j\neq i }  \delta^{i,j}(u, Y_{u-}) d M^{i,j}_u).
\]
So, we obtain that the  predictable processes $\delta$, $J(\cdot,x)$, $\gamma_{i,j}$ in the representation \eqref{eq:repr-H-P}  are given by
\begin{align}\label{eq:mart-rep-d}
\delta_t &=
\sum_{i \in \mathcal{K}} \frac{H^i_{t-}}{B_t}\left[ (\sigma_{\!S}(t,Y_{t-},
i))^\top (\nabla_{\!\!S} v(t,Y_{t-},i))^\top + (\sigma_{\!R}(t,Y_{t-},
i))^\top (\nabla_{\!\!R} v(t,Y_{t-},i))^\top \right],
\\
\nonumber
J_t(x) &= \sum_{i \in \mathcal{K}} \frac{H^i_{t-}}{B_t}( v (t , Y_{t-} +
F_Y(t,Y_{t-},  i,x),i ) - v(t, Y_{t-},i)),
\\
\nonumber
\gamma^{i,j}_t &= B^{-1}_t\left(v(t, Y_{t-} + \rho^{i,j}_Y(t,Y_{t-}),j) - v(t,
Y_{t-},i) + \delta^{i,j}(t,Y_{t-})\right).
\end{align}
Thus the assertion follows from Theorem \ref{thm:risk-min-spec}.
\end{proof}
\begin{rem}
If in SDE \eqref{eq:SDE-SR} there is no component $R$ (i.e., $Y=S$) then, by Remark \ref{rem:hedging}.c and \eqref{eq:mart-rep-d}, the Brownian-risk can be eliminated by  a usual delta-hedging strategy which is given
by
\[
	\phi_t =(\nabla_{\!\!S} v(t,S_{t-},C_{t-}))^\top.
\]
If component $R$ appears in diffusion model, i.e., in \eqref{eq:SDE-SR} with $F_S=0$, $F_R = 0$, $\rho_S=0$ and  $\rho_R =0$,  then the Brownian-risk cannot be eliminated completely. The risk-minimizing strategy is given by
\[
	\phi_t =(\nabla_{\!\!S} v(t,Y_{t-},C_{t-}))^\top + (a^{-1}_{SS}a_{SR})(t,Y_{t-},C_{t-})(\nabla_{\!\!R} v(t,Y_{t-},C_{t-}))^\top,
\]
where $a^{-1}_{SS}$ denotes the Moore-Penrose pseudo-inverse of $a_{SS}$ (see Albert \cite{Albert1972}).
Hence, if $a_{SR}=0$, i.e., if continuous martingale parts of  $R$ and $S$ are orthogonal, then
  $\phi_t =(\nabla_{\!\!S} v(t,Y_{t-},C_{t-}))^\top$, and the risk-minimizing strategy is a usual delta hedging strategy.

\emph{Delta-hedging of credit risk} strategy is \mn{$\phi(t) = \sum_{i \in  \mathcal{K}} \I_i(C_{t-}) \phi^i_t$, where $\{ \phi^i \}_{ i \in \mathcal{K} }$} are appropriately measurable solution to the  system of linear equations:
\[
\left(
  \begin{array}{cccc}
    \rho^{i,1,1}_S & \rho^{i,1,2}_S & \cdots &  \rho^{i,1,d}_S \\
    \rho^{i,2,1}_S & \rho^{i,2,2}_S & \cdots  & \rho^{i,2,d}_S \\
    \vdots & \vdots & \ddots &  \vdots \\
    \rho^{i,i-1,1}_S & \rho^{i,i-1,2}_S & \cdots  & \rho^{i,i-1,d}_S \\
    \rho^{i,i+1,1}_S & \rho^{i,i+1,2}_S & \cdots  & \rho^{i,i+1,d}_S \\
    \vdots & \vdots & \cdots &  \vdots \\
    \rho^{i,K,1}_S & \rho^{i,K,2}_S & \ddots &  \rho^{i,K,d}_S \\
  \end{array}
\right)
\phi^i_t =
	\left(
   \begin{array}{c}
     \Delta_{i,1}v  + \delta_{i,1}\\
     \Delta_{i,2}v  + \delta_{i,2}\\
     \vdots \\
     \Delta_{i,i-1}v+ \delta_{i,i-1}\\
     \Delta_{i,i+1}v+ \delta_{i,i+1}\\
     \vdots \\
     \Delta_{i,K}v+ \delta_{i,K}\\
   \end{array}
 \right),
\]
where
\[
\Delta_{i,j}v(t,Y_{t-})
:= v(t, Y_{t-} + \rho^{i,j}_Y(t,Y_{t-}),j) - v(t,
Y_{t-},i),
\]
and $\rho^{i,j,k}_S(t,Y_{t-}) $ denotes the $k$-th coordinate of $\rho^{i,j}_S(t,Y_{t-})$.
For convenience we drop the dependence on $t$ and $Y_{t-}$ in the above system of equations.
\end{rem}
\begin{rem}\label{rem:mom-est}

a) The  conditions ensuring that \eqref{eq:sup-S2} holds are  (LG) together with  (LGm), i.e.
%
  \begin{align} \label{LGm}
    \int_\mathbf{R^n}|F_Y(u,z,i,x)|^{2m} \nu(dx)
    \leq
    K(1 + |z|^{2m}),
  \end{align}
    for some $m \geq 1$,
  see \cite[Thm. 3.2]{jaknie2012}.
  Note that if there exists a function $\mathcal{K}$ with
  \[
   \int_{\mathbf{R^n}} (|\mathcal{K}(x)|^2 + |\mathcal{K}(x)|^{2m})
   \nu(dx) < \infty
  \]
  for some $m \geq 1$,
  such that
  \[
  	\frac{|F_Y(u,z,i,x)|}{1 + |z|} \leq \mathcal{K}(x) \quad \forall z \in \mathbf{R}^{d+p},
  \] then $F$ satisfies (LG) and (LGm).

b) Note that if (LG) and \eqref{LGm} hold, and  $v \in \mathcal{C}^{1,2}$ is such that
$|\nabla^2 v(u,z,i)| \leq K (1 + |z|^{2m-2})$ for each $i \in \mathcal{K}$,
 then   $v \in \mathcal{C}^{1,2} \cap \mathcal{I}$.
\end{rem}

\section{Feynman-Kac Theorem and PIDE}

In this section we formulate and prove  Feynman-Kac Theorem for components  $(Y^{t,y,c}_u,C^{t,y,c}_u)_{u \in [\![t,T ]\!]}$ of a weak solution to SDE \eqref{eq:SDE-gen}. So, we connect calculation of conditional expectation of some interesting and important functionals of $(Y^{t,y,c}_u,C^{t,y,c}_u)_{u \in [\![t,T ]\!]}$ with solution of corresponding PIDE with given boundary conditions. We consider functionals of the form \eqref{eq:D}
with $h$, $g$, $\delta^{i,j}$ satisfying \eqref{eq:D1}.
With such functional we connect a process $V$ given by \eqref{eq:proc-ex-div}, and for $V$ we prove a Feynman-Kac type theorem.
In general context our results concerns functionals of the type
\begin{align*}
& e^{-\int_t^T \!\mathfrak{r}( v, Y^{t,y,c}_{v}, C^{t,y,c}_{v}  )dv }h(Y^{t,y,c}_T, C^{t,y,c}_T) +
\!\!\int_{t}^T \!\! e^{-\int_t^u \!\mathfrak{r}( v, Y^{t,y,c}_{v}, C^{t,y,c}_{v}  )dv}g(u,Y^{t,y,c}_{u-}, C^{t,y,c}_{u-} ) du
\\
\nonumber
&
+
\!\!\sum_{i,j\neq i } \int_{t}^T \!\!e^{-\int_t^u \!\mathfrak{r}( v, Y^{t,y,c}_{v}, C^{t,y,c}_{v}  )dv} \delta^{i,j}(u,
Y^{t,y,c}_{u-}) d H^{i,j}_u.
\end{align*}
Note also that if there is no component $C$ we have classical Feynman-Kac Theorem for jump-diffusion with killing at rate $\mathfrak{r}$.
Thus, we solve a general problem which is of independent interest and has many applications, among others in finding risk-minimizing strategies, which  will be  presented in this section.

The process $A$ given by \eqref{eq:A} is equal to zero as we see in the proof of Theorem \ref{thm:risk-min-markov}.
By right continuity of paths of $(Y,C)$ we obtain
\[
L_t := \int_0^t \frac{1}{B_u}\!\left( \!( \partial_t + \mathcal{A}_u - \mathfrak{r}) v(u,Y_{u},C_{u}) + g(u,Y_{u},C_{u}) +\sum_{j \in \mathcal{K} \setminus C_{u} } \delta^{C_{u},j}(u,Y_{u}) \lambda^{C_{u},j}(u,Y_{u}) \!\right)\!du
= 0.
\]
The mean value theorem for right continuous functions yields
\begin{align*}
 ( \partial_t + \mathcal{A}_0 - \mathfrak{r}) v(0,y,c) + g(0,y,c)
 +\sum_{j \in \mathcal{K} \setminus c } \delta^{c,j}(0,y) \lambda^{c,j}(0,y) = 0.
\end{align*}
Using similar arguments we obtain
\begin{thm}\label{thm:FK-nec}
Assume EUWS and that \eqref{eq:sup-S2} and \eqref{eq:LG-D} hold for some $m \geq 1 $. Let $v$ be a function defined by
\eqref{eq:v-def}. If $v \in \mathcal{C}^{1,2} \cap \mathcal{I}$, then
$v$ solves the following PIDE 	
\begin{align}\label{eq:PIDE}
 ( \partial_t + \mathcal{A}_t - \mathfrak{r}) v(t,y,c) + g(t,y,c) +\sum_{j \in \mathcal{K} \setminus c } \delta^{c,j}(t,y) \lambda^{c,j}(t,y)
= 0
\end{align}
on  $ ]\!] 0,T[\![ \times  \mathbf{R}^{d+p} \times \mathcal{K} $ 
with boundary condition
\begin{equation} \label{eq:PIDE-boundary}
    v(T,y,c) = h(y,c).
\end{equation}
\end{thm}

This gives us a way to finding the ex-dividend price function $v$ and to show connection with so called Feynman-Kac type theorems. Thus to obtain $v$ we solve PIDE \eqref{eq:PIDE} and if solution is a sufficiently regular function, then $v$ has the stochastic representation
\[
v(t,Y_t,C_t) =
\!B_t\mathbb{E} \!\!\left( \frac{\!h(Y_T, C_T)}{B_T} +
\!\!\int_{t}^T \!\!\frac{g(u,Y_{u-}, C_{u-} )}{B_u} du +
\!\!\sum_{i,j\neq i } \int_{t}^T \!\!\!\frac{\delta^{i,j}(u,
Y_{u-}) }{B_u}d H^{i,j}_u \bigg| \mathcal{F}_t
\!\!\right)\!.
\]
So $v$ is indeed the ex-dividend price function.
The next theorem gives precise formulation of this idea.
\begin{thm}\label{thm:Fey-Kac}
Let  $(Y^{t,y,c}_u,C^{t,y,c}_u)_{u \in [\![t,T ]\!]}$ be the components of a weak solution to SDE \eqref{eq:SDE-gen} started at $t$ from $(y,c) \in \mathbf{R}^{d+p} \times \mathcal{K}$.
 Assume that a function $v \in \mathcal{C}^{1,2} \cap \mathcal{I}$ solves the PIDE \eqref{eq:PIDE} with the boundary
condition \eqref{eq:PIDE-boundary} and
the local martingale $\widehat{M}$ on $[\![t,T]\!]$ with the dynamic
    \begin{align}
\nonumber d \widehat{M}_u &= \sum_{i} H^i_{u- } \nabla v(u,
Y^{t,y,c}_{u-},i) \sigma_Y(u,Y^{t,y,c}_{u-}, i) d W^{t,y,c}_u
\\
\label{eq:mart-dM} & + \sum_{i} H^{i}_{u-} \int_{\mathbf{R}^n}
(v(u , Y^{t,y,c}_{u-} +  F_Y(u,Y^{t,y,c}_{u-}, i, x ),i) - v(u , Y^{t,y,c}_{u-},i))
\widetilde{\Pi }^{t,y,c}(dx,du)
\\
\nonumber & + \sum_{i,j \neq i } (v(u , Y^{t,y,c}_{u-} + {\rho^{i,j}_Y(u,
Y^{t,y,c}_{u-})},j) - v(u , Y^{t,y,c}_{u-},i) + \delta^{i,j}(u,Y^{t,y,c}_{u-})) d(M^{i,j})^{t,y,c}_u
\\
\nonumber \widehat{M}_t &= v(t,y , c )
\end{align}
is  a martingale on $[\![t,T]\!]$. Then
\begin{align}\label{eq:Fey-Kac}
v(t,y,c)
= \!\mathbb{E}^{t,y,c} \!\bigg(
& \frac{h(Y^{t,y,c}_T, C^{t,y,c}_T)}{B^{t,y,c}_{T}   } +
\!\!\int_{t}^T \!\! \frac{g(u,Y^{t,y,c}_{u-}, C^{t,y,c}_{u-} ) }{B^{t,y,c}_{u} }du
\\
\nonumber
&
+
\!\!\sum_{i,j\neq i } \int_{t}^T \!\!\! \frac{\delta^{i,j}(u,
Y^{t,y,c}_{u-}) }{B^{t,y,c}_{u} }d H^{i,j}_u .
\end{align}
\end{thm}
\begin{proof}
 Let $v$ be a function  satisfying PIDE
\eqref{eq:PIDE} with the boundary condition
\eqref{eq:PIDE-boundary}. For $v \in
\mathcal{C}^{1,2} \cap \mathcal{I} $
using the   It\^o lemma (in a similar way as in \cite[Thm. 2.4]{jaknie2012}), we
have
\begin{align}\label{eq:vTminus}
    \frac{v(T,Y^{t,y,c}_{T-},C^{t,y,c}_{T-})}{B^{t,y,c}_{T}} &=  v(t,y,c)
    +
\int_{t}^T
\!\!\!
\frac{1}{B^{t,y,c}_{u}}
\left( \partial_t+\mathcal{ A}_u -\mathfrak{r}\right)v(u,Y^{t,y,c}_{u-} ,C^{t,y,c}_{u-})  du
 + \int_{t}^T 
\frac{1}{B^{t,y,c}_{u}}
 dM^v_u,
\end{align}
where $M^v$ is given,  for $w \geq t$, by
\begin{align*}
& M^v_w := \int_t^w \sum_{i} H^i_{u- } \nabla
v(u, Y^{t,y,c}_{u-},i)\sigma_Y(u,Y^{t,y,c}_{u-}, i) d W^{t,y,c}_u
\\
\nonumber
& + \int_t^w \sum_{i}  H^i_{u- }
 \int_{\mathbf{R}^n} (v(u , Y^{t,y,c}_{u-} +  F_Y(u,Y^{t,y,c}_{u-}, i, x ),i) - v(u , Y^{t,y,c}_{u-},i)) \widetilde{\Pi}^{t,y,c} (du,dx)
\\
\nonumber & +  \int_t^w \sum_{i,j \neq i } (v(u , Y^{t,y,c}_{u-} +
{\rho^{i,j}_Y}(u,Y^{t,y,c}_{u-}),j) - v(u , Y^{t,y,c}_{u-},i) ) d (M^{i,j})^{t,y,c}_u.
\end{align*}
We have  $v(T,Y^{t,y,c}_{T},C^{t,y,c}_{T}) =  v(T,Y^{t,y,c}_{T-},C^{t,y,c}_{T-})$, since  $\Delta N^{i,j}_T =0$ for every $i,j \in \mathcal{K}$, $i \neq j$  and $\Pi(A \times \set{T}) = 0$ for every $A \in \mathcal{B}(\mathbf{R}^n)$.
%
Hence, by  \eqref{eq:vTminus} and  assumption that $v$ solves PIDE \eqref{eq:PIDE} with boundary condition \eqref{eq:PIDE-boundary}, we get
\begin{align*}
    \frac{h(Y^{t,y,c}_T,C^{t,y,c}_T)}{B^{t,y,c}_{T}}&+
\int_{t}^T
\!\!\!
\frac{1}{B^{t,y,c}_{u}}
\bigg(
g(u,Y^{t,y,c}_u,C^{t,y,c}_{u})
+   \sum_{\substack{i,j \in \mathcal{K} \\ j \neq i} }
H^i_{u- }
 \delta^{i,j}(u, Y^{t,y,c}_{u-}) \lambda^{i ,j} (u,Y^{t,y,c}_{u-})\bigg)
du \\
&
=
 v(t,Y^{t,y,c}_t,C^{t,y,c}_t)+  \int_{t}^T \!\!\!
\frac{1}{B^{t,y,c}_{u}}
 dM^v_u.
\end{align*}
Thus, using \eqref{eq:M-ijtyc}
, after  rearranging,
we obtain that this
equality takes the form
\begin{align}
\label{eq:v-1} &&     \frac{h(Y^{t,y,c}_T,C^{t,y,c}_T)}{B^{t,y,c}_{T}}
  + \int_{t}^T \!\! \frac{ g(u,Y^{t,y,c}_{u},C^{t,y,c}_{u}) }{B^{t,y,c}_{u}} du +   \sum_{i,j \neq i}
\int_t^T \frac{ \delta^{i,j}(u, Y^{t,y,c}_{u-}) }{B^{t,y,c}_{u}} d (H^{i,j})^{t,y,c}_u
\\
\nonumber
 &&=
 v(t,y,c) +  \int_{t}^T \!\!\!
\frac{1}{B^{t,y,c}_{u}}
 \left(dM^v_u
+   \sum_{i,j \neq i} \delta^{i,j}(u, Y^{t,y,c}_{u-})
  d (M^{i,j})^{t,y,c}_u \right).
\end{align}
By assumption,
\[
\widehat{M}_w = M^v_w +   \sum_{i,j \neq i}
\int_t^w
\delta^{i,j}(u, Y^{t,y,c}_{u-})
 d (M^{i,j})^{t,y,c}_u
 \]
  is  a martingale, so taking
the expectation in \eqref{eq:v-1} we obtain
\begin{align*}
\mathbb{E}^{t,y,c}  \left(
\frac{h(Y^{t,y,c}_T,C^{t,y,c}_T)}{B^{t,y,c}_{T}}
  + \int_{t}^T
    \!\!
    \frac{g(u,Y^{t,y,c}_{u},C^{t,y,c}_{u})}{B^{t,y,c}_{u}} du +   \sum_{i,j \neq i}
\int_t^T   \!\! \frac{\delta^{i,j}(u, Y^{t,y,c}_{u-}) }{B^{t,y,c}_{u} } d (H^{i,j})^{t,y,c}_u
\right)
\!=
 v(t,y,c).
\end{align*}
  This finishes the proof.
\end{proof}
In Theorem \ref{thm:Fey-Kac} above we assume that corresponding PIDE has classical solutions. In a very special case of our setting we can find sufficient conditions for this assumption to hold, e.g. see Gichman Skorochod \cite[Theorem II.2.6 and Corolary II.2.2]{gichskorIII} (see also our example \ref{sec:gen-exp-reg}).
The question about the form of sufficient conditions for existence of classical solutions to our PIDE is very interesting but it requires deeper study and therefore it is postponed to the following paper.

 Assumption that the local martingale \eqref{eq:mart-dM} is a martingale is the weakest that one can formulate to obtain stochastic representation of classical solution of PIDE.
 In the theorem below we formulate some sufficient conditions for the corresponding martingale property.
Let $\mathcal{C}_m$ be a class of functions having polynomial growth of order $m$ defined  by
\begin{align*}
\mathcal{C}_m=\left\{  v: [\![0,T]\!] \right. & \times \mathbf{R}^{d+p} \times
\mathcal{K} \rightarrow \mathbf{R}: \ v(\cdot, \cdot, i) \text{ is continuous and }  \\
 & \left.|v(u,z,i)| \leq K (1 + |z|^m) \quad  \forall i \in \mathcal{K} \right\},
\end{align*}
and $\mathcal{C}^{1,2,m}$  be  the  space of all measurable functions $v
: [\![0,T]\!] \times \mathbf{R}^d \times \mathcal{K} \rightarrow
\mathbf{R} $ such that $ v(\cdot, \cdot, i ) \in
C^{1,2}([\![0,T]\!] \times \mathbf{R}^d  )$ for every $  i \in
\mathcal{K} $ and with the derivative with respect to $z$ having a
polynomial growth of degree $m$, i.e., $ |\nabla v(u,z,i)| \leq K
(1 + |z|^m)$.
\begin{thm}\label{thm:Fey-Kac-nowy}
	Let  $(Y^{t,y,c}_u,C^{t,y,c}_u)_{u \in [\![t,T ]\!]}$ be the components of a solution to SDE \eqref{eq:SDE-gen} started at $t$ from $(y,c) \in \mathbf{R}^{d+p} \times \mathcal{K}$.
    Assume (LG), (LGm) for some natural number
    $m \geq 1$, and
    \begin{align}\label{eq:LG-delta-m}
    \left|\delta^{i,j}(u,z) \right|^2  \leq K(1 + |z|^{2m}),
    \quad
    \forall i,j \in \mathcal{K}, i \neq j.
    \end{align}
    If $v \in \mathcal{C}^{1,2,m-1} \cap \mathcal{C}_m$ solves PIDE \eqref{eq:PIDE} with boundary condition \eqref{eq:PIDE-boundary}, then \eqref{eq:Fey-Kac} holds.
\end{thm}
\begin{proof}

Fix  $t,y,c$.
 From Theorem \ref{thm:Fey-Kac} follows that it is sufficient to
show that each stochastic integral in formula \eqref{eq:mart-dM}, defining the local martingale $\widehat{M}$, belongs to $\mathcal{H}^2$ - the space of square integrable martingales.
We start from the integral with
respect to a Brownian motion, i.e.,
\[
	\widehat{M}^c_w :=\int_t^w \sum_{i} H^i_{u- } \nabla
v(u, Y^{t,y,c}_{u-},i)\sigma_Y(u,Y^{t,y,c}_{u-}, i) d W_u.
\]
Since $v \in C^{1,2,m-1} $ we infer, using (LG), that
\begin{align*}
\left| \nabla  v(u, z ,i) \sigma_Y (u, z , i) \sigma_Y(u, z , i)^\top
\left( \nabla v(u, z ,i) \right)^\top \right|^2
& \leq |\nabla v(u, z ,i)|^2 \norm{\sigma_Y(u, z , i) \sigma_Y(u,z,i)^\top }^2
\\
& \leq L_1 (1 + |z|^{2m-2})(1 + |z|^{2})
\\
& \leq L_2 (1 + |z|^{2m}) .
\end{align*}
This and \eqref{LGm} imply, by Remark \ref{rem:mom-est}, that
\begin{align*}
\mathbb{E}^{t,y,c}
\int_t^T  \bigg|
 \nabla  v(u, Y^{t,y,c}_{u-},i)
a_{YY}(u,Y^{t,y,c}_{t-}, i )& \left( \nabla
v(u, Y^{t,y,c}_{u-},i) \right)^\top \bigg|^2 du
\\
&
\leq
L_2 T
\mathbb{E}^{t,y,c}
\left(1+ \sup_{u \in [\![t,T]\!]} |Y^{t,y,c}_u|^{2m} \right) < \infty . 
\end{align*}
Hence, by the Doob $L^2$-inequality for local martingales applied to $M^c$ we obtain that $M^c \in \mathcal{H}^2$.
Now, we consider the second integral in \eqref{eq:mart-dM}, i.e.,
 \[
\widehat{M}^{d1}_w:=
\int_t^w \sum_{i}  H^i_{u- }
 \int_{\mathbf{R}^n} (v(u , Y^{t,y,c}_{u-} +  F_Y(u,Y^{t,y,c}_{u-}, i, x ),i) - v(u , Y^{t,y,c}_{u-},i)) \widetilde{\Pi} (du,dx).
\]
Obviously  $\widehat{M}^{d1} \in \mathcal{M}_{loc}$.
Let $[\widehat{M}^{d1}]$ be a process of quadratic variation of $\widehat{M}^{d1}$. We  have
\[
\left[\widehat{M}^{d1} \right]_T=
\int_t^T \sum_{i}  H^i_{u- }
 \int_{\mathbf{R}^n} |v(u , Y^{t,y,c}_{u-} +  F_Y(u,Y^{t,y,c}_{u-}, i, x ),i) - v(u , Y^{t,y,c}_{u-},i)|^2 \Pi (du,dx)
\]
and, as before, we use the Doob inequality to show that $\widehat{M}^{d1} \in \mathcal{H}^2$. So, it is suffices to show that $\mathbb{E}^{t,y,c} [\widehat{M}^{d1}]_T < \infty$.
Since
$v \in C^{1,2,m-1}$, we have
\begin{align*}
|v(u,z +  F_Y(u,z,i,x),i) - v(u,z,i))|
&\leq | \nabla v(u,z + \theta^{\!\top} F_Y(u,z,i,x),i )| |F_Y(u,z,i,x)|
\\
& \leq L_3 | 1 + | z + \theta^{\!\top} F_Y(u,z,i,x) |^{m-1} | |F_Y(u,z,i,x)|
\\
& \leq  L_3 | 1 + |z|^{m-1}||F_Y(u,z,i,x)| + L|F_Y(u,z,i,x)|^{m}.
\end{align*}
Hence, by (LG) and (LGm)
\begin{align*}
&\int_{\mathbf{R}^n}
|v(u,z +  F_Y(u,z,i,x),i) - v(u,z,i))|^2\nu (dx)
\\
&
\leq L_4 ( 1 + |z|^{2m-2})\int_{\mathbf{R}^n}|F_Y(u,z,i,x)|^2\nu (dx)
+ L_4 \int_{\mathbf{R}^n}|F_Y(u,z,i,x)|^{2m}\nu (dx)
\\
&\leq L_5 K ( 1 + |z|^{2m}).
\end{align*}
Using this estimation we obtain
\begin{align*}
& \mathbb{E}^{t,y,c}\int_t^T
 \int_{\mathbf{R}^n} |v(u , Y^{t,y,c}_{u-} +  F(u,Y^{t,y,c}_{u-}, i, x ),i) - v(u , Y^{t,y,c}_{u-},i)|^2 \Pi (du,dx)
\\
 &= \mathbb{E}^{t,y,c}\int_t^T
 \int_{\mathbf{R}^n} |v(u , Y^{t,y,c}_{u-} +  F(u,Y^{t,y,c}_{u-}, i, x ),i) - v(u , Y^{t,y,c}_{u-},i)|^2 \nu (dx) du
 \\
&\leq
L_6 \mathbb{E}^{t,y,c}
\left( 1 + \sup_{u \in [\![t,T]\!]}| Y^{t,y,c}_u|^{2m}\right)
<\infty.
\end{align*}
Therefore $\widehat{M}^{d1} \in \mathcal{H}^2$. Now we consider the third part, i.e.,
\begin{align*}
\widehat{M}^{d2}_w:= & \int_t^w \sum_{i,j \neq i } (v(u , Y^{t,y,c}_{u-} + {\rho^{i,j}_Y(u,
Y^{t,y,c}_{u-})},j) - v(u , Y^{t,y,c}_{u-},i) + \delta^{i,j}(u,Y^{t,y,c}_{u-})) d
M^{i,j}_u.
\end{align*}
We obviously have
\begin{align*}
&\left|v(u,z + {\rho^{i,j}_Y(u, z )},j) - v(u,z,i) +
\delta^{i,j}(u,z)\right|^2
\\
&
\leq L_7  \left(
\left|v(u,z + {\rho^{i,j}_Y(u, z )},j) - v(u,z, j)\right|^2
	+
\left|v(u,z, j) - v(u,z, i)\right|^2
+
\left|\delta^{i,j}(u,z)\right|^2 \right).
\end{align*}
The similar arguments as above, assumption (LG) and $v \in \mathcal{C}^{1,2,m-1}$ yield
\begin{align*}
\left|v(u,z + {\rho^{i,j}_Y(u, z )},j) - v(u,z, j)\right|^2
\leq L_8 (1 + |z|^{2m}).
\end{align*}
This,  \eqref{eq:LG-delta-m} and the fact that $v \in \mathcal{C}^{1,2,m-1} \cap \mathcal{C}_m$ give
\begin{align*}
\left|v(u,z + {\rho^{i,j}_Y(u, z )},j) - v(u,z, i) +
\delta^{i,j}(u,z)\right|^2
\leq L_9 (1 + |z|^{2m}).
\end{align*}
Using boundedness of
$\lambda^{i,j}$ and repeating arguments we get
\[
\mathbb{E}^{t,y,c} \int_t^T \!\!\! \left|v(u , Y^{t,y,c}_{u-}\!\!+\!
{\rho^{i,j}_Y(u, Y^{t,y,c}_{u-})},j) \!-\! v(u , Y^{t,y,c}_{u-},i) \!+\!
\delta^{i,j}(u,Y^{t,y,c}_{u-})\right|^2 H^{i}_{u-}
\lambda^{i,j}(u,Y^{t,y,c}_{u-})du < \infty,
\]
so  $\widehat{M}^{d2}$ is in $\mathcal{H}^2$, see \cite[Thm. I.4.40]{js1987}.
This finishes the proof.
\end{proof}
\section{Further generalizations and examples}
\subsection{Further generalizations}
Results of the previous sections can be extended in several directions.
One possibility is to assume that $\widetilde{\Pi}(dx,du)$ is a  compensated integer-valued random measure on
$\mathbf{R}^n \times [\![0,T]\!]$ with the intensity measure given by
\[
Q(dx,du)=K(u,Y_{u-},C_{u-},x ) \bar\nu(dx) du
\]
 where $\nu$ is a L\'{e}vy measure, 
and $K$ is nonnegative real-valued measurable function.
For such compensator we impose the following growth type conditions:

\textbf{Assumption GC-$K$.}
Functions $\sigma$, $F$, $\rho^{i,j}$  satisfy the linear growth condition
  \begin{align*}\tag{LG-$K$}
  	|\sigma(t,y,c)|^2
  +
  \int_\mathbf{R^n}|F_Y(u,z,i,x)|^2 K(u,z,i,x)\bar \nu(dx)
  +
  \sum_{j:j \neq i}
  |\rho^{i,j}_Y(u,z)|^2
  \leq L(1 + |z|^2),
  \end{align*}
  and the polynomial growth condition
  \begin{align*}\tag{LGm-$K$}
    \int_\mathbf{R^n}|F_Y(u,z,i,x)|^{2m} K(u,z,i,x)\bar \nu(dx)
    \leq
    L(1 + |z|^{2m}),
  \end{align*}
  for some constant $L > 0$ and $m \geq 1$.

Under these conditions we have moment estimate \eqref{eq:sup-S2} which can be proved analogously as in \cite{jaknie2012} (c.f., Remark \ref{rem:mom-est}a). Thus
assuming EUWS we can easily generalize results of previous section.
For example, family of generators $\mathcal{A}_t$ should be replaced by $\mathcal{A}^K_t$ defined by
\eqref{eq:gener-SDE}  with $K(u,z,i,x)\bar \nu(dx)$ instead of $\nu(dx)$.
 The  corresponding results on risk-minimization and its connection with PIDE's can be proved analogously. The only difference is in formulations in which we have to  replace $\nu(dx)$ with $K(u,z,i,x)\bar \nu(dx)$ and (LG), (LGm) with (LG-$K$), (LGm-$K$) respectively. There is a nontrivial question about conditions which ensure assumption EUWS, i.e., about the existence of unique weak solutions to SDE in such case. This question is very interesting, but we postpone it to another paper.

\subsection{Examples}
The setup considered in this paper is very general, so it contains many models well known in finance such as local volatility models and regime switching L\'{e}vy models. In this subsection we present examples which demonstrate how our theory  generalize known results in existing models.
\subsubsection{ General exponential L\'{e}vy model with stochastic volatility }
In many practical applications we are interested in models with nonnegative price processes. In majority of financial applications this requirement is achieved by imposing a linear structure on the coefficients in dynamic of $S$ (i.e. the first $d$ components of $Y$ in our model),  by taking \begin{align*}
d S^k_t &= S^k_{t-} \bigg( \mathfrak{r}(t,R_{t-},C_{t-}) du + \Sigma_k(t,R_{t-},C_{t-})\!^\top\! d W_t + \!\!\int_{\mathbf{R}^n}\!\! (e^{\Sigma_k(t,R_{t-},C_{t-})\!^\top\! x} -1) \widetilde{\Pi}( dx, dt) \\
	&\quad\quad\quad   + \!\!\!\!\sum_{i,j \in
    \mathcal{K}: j \neq i } \!\!\!\!(e^{P^{i,j}_k(t,R_{t-})} - 1 )  d M^{i,j}_t \bigg),
\\
d R_t &= \mu_R(t,R_{t-},C_{t-}) du + \sigma_R(t,R_{t-},C_{t-})\!^\top\! d W_t,
\end{align*}
where $W, \widetilde{\Pi}, M^{i,j}$ are as before with $\lambda^{i,j}$ being deterministic function. In this model we assume, for simplicity, that the coefficient of diffusion of $R$ is modulated by $C$ - the process  given in advance. We assume EWUS. This is provided if we  assume that $\Sigma_k : [\![0,T]\!] \times \mathbf{R}^p \times \mathcal{K} \rightarrow \mathbf{R}^n$  and  $P^{i,j}_k : [\![0,T]\!] \times \mathbf{R}^p \rightarrow \mathbf{R}$ are continuous
and bounded functions for all $k \in \set{1, \ldots d}$,  $\mu_R$, $\sigma_R$ satisfy linear growth and local Lipschitz conditions and
\[
\Sigma_k(t,r,c) + \Sigma_l(t,r,c) \in B
\quad
\forall k,l \in \set{1,\ldots, n},
 \]
 for all
\[
	B=\set{ v \in \mathbf{R}^{n}:     \int_{|x|>1}e^{ v^\top \!x} \nu(dx) <
    \infty
}.
\]

Let us denote by $\text{Exp}(a)$ a component-wise exponential function of vector $a$, by $\I_d$ a d-dimensional column vector of $1$.
Assume that the dividend process $D$ is given by \eqref{eq:D}
and the value function is appropriately smooth. Using Theorem \ref{thm:risk-min-markov} we obtain  that  the component $\phi$  of a  risk-minimizing hedging strategy has, on the set $\set{C_{t-} = i}$, the form
\begin{align*}
\phi_t &= \diag(S_{t-})^{-1} (\widehat{G}_t(R_{t-}, i) )^{-1} \\
& \bigg(\Sigma^\top\Sigma( t,R_{t-},i)\diag(S_{t-})  \nabla\!_{S} v( t,Y_{t-},i ) + \Sigma^\top\sigma_R( t,Y_{t-},i)  \nabla\!_{R} v( t,Y_{t-},i )
\\
&+ \int_{\mathbf{R}^n}
\left(\text{Exp}(\Sigma( t,R_{t-},i)\!^\top\!x)-\I\right)
\big(v( t,S_{t-} \!\circ\!\text{Exp}(\Sigma( t,R_{t-},i)\!^\top\!x),R_{t-}, i) - v( t,S_{t-},R_{t-}, i)\big)\nu(dx) \\
&+
\!\!\!\!
  \sum_{j \in \mathcal{K} : j \neq i} \!\!\! \left(\text{Exp}(P^{i,j}(t,R_{t-}))-\I\right)
   \big(v( t,S_{t-}\!\circ\!\text{Exp}(P^{i,j}(t,R_{t-}) ,R_{t-},j) - v( t,S_{t-},R_{t-},i ) \\
&
+ \delta^{i,j}(t,Y_{t-})\big){\lambda^{i,j}(t)}
\bigg),
\end{align*}
where
\begin{align*}
\widehat{G}_t(r,i) &= \bigg(
\Sigma^\top\Sigma( t,r,i)
+\!\!\int_{\mathbf{R}^n}
\!\!\!\!
\left(\text{Exp}(\Sigma( t,r,i)\!^\top\!x)-\I_d\right)\!\!
\left(\text{Exp}(\Sigma( t,r,i)\!^\top\!x)-\I_d\right)\!^\top\nu(dx) \\
&
\quad\quad\quad\quad\quad\;\; + \!\!\!\!\sum_{j \in K, j\neq i}
\!\!\!\!\left(\text{Exp}(P^{i,j}(t,r))-\I_d\right)\!
\left(\text{Exp}(P^{i,j}(t,r))-\I_d\right)\!^\top \lambda^{i,j}(t)
\bigg),
\end{align*}
and $(\widehat{G}_t(R_{t-}, i) )^{-1}$ is the Moore-Penrose pseudo-inverse, and  $\circ$ denotes the Hadamard product (i.e. the componentwise product of matrices).

\subsubsection{Generalization of exponential L\'{e}vy model with regime-switching }\label{sec:gen-exp-reg}

Now we consider particular model of asset prices and present how our theory works.
We consider market with money account and one risky asset.
 The
dynamic of asset price process is given by
\begin{align}\label{eq:SDE-exp-levy}
    d S_t &= S_{t-}\!\!\left( \mathfrak{r}(C_{t-}) dt + \sigma(C_{t-})\!^\top\! d W_t  + \!\!\int_{\mathbf{R}^n}\!\!
    (e^{\sigma(C_{t-})\!^\top\!x } - 1 ) \widetilde{\Pi}( dx, dt)+ \!\!\!\!\sum_{i,j \in
    \mathcal{K}: j \neq i } \!\!\!\!(e^{\rho^{i,j}} - 1 )  d M^{i,j}_t\right), \\
    \nonumber
    d C_t &= \sum_{i,j \in \mathcal{K} : j \neq i} (j-i)\I_\set{ i} (C_{t-}) d
    N^{i,j}_t ,
\end{align}
where $\sigma: \mathcal{K} \rightarrow \mathbf{R}^n$, $\rho^{i,j}
\in \mathbf{R}$, $N^{i,j}$ are independent Poisson processes with
constant intensities $\lambda^{i,j} >0 $, and $\Pi(dx,dt)$ is a
Poisson random measure with intensity measure $\nu(dx)dt$ satisfying, for some $m \geq 1$,
\begin{equation}\label{eq:sigma-2m}
    \int_{|x|>1}e^{2 m\sigma(i)\!^\top\!x} \nu(dx) <
    \infty
    \qquad
    \forall
    i \in \mathcal{K} .
\end{equation}
Our model generalize a regime switching model with jumps, extensively studied amongs others by Chourdakis \cite{ch2005}, Mijatovic and Pistorius \cite{mijpis2011} or
Kim et.al. \cite{kflr2012}  for which $\rho^{i,j} = 0$  in  \eqref{eq:SDE-exp-levy}.
Note that the coefficients of  SDE \eqref{eq:SDE-exp-levy} satisfy standard assumptions for existence of a unique strong solution, and
 the  coordinate $C$ of solution $(S,C)$ is a Markov chain
with the state space $\mathcal{K}$.
Theorem \ref{thm:risk-min-spec} yields that the component $\phi$ of the risk-minimizing strategy
for a  dividend process $D$ with representation \eqref{eq:repr-H-P} is given, on the set $\set{C_{t-}=i}$, by
\begin{equation*}\label{eq:LS-cond-ort-1D}
\phi_t = \frac{\left( \sigma(i)\!^\top\!\delta_t + \int_{\mathbf{R}^n}
(e^{\sigma(i)\!^\top\!x} -1) J_t(x) \nu(dx) +
  \sum_{j \in \mathcal{K} : j \neq i} (e^{\rho^{i,j}}-1)  \gamma^{i,j}_t  \lambda^{i,j}
\right)}{S_{t-}\widehat{G}^{i}_t},
\end{equation*}
where
\begin{equation*}
\widehat{G}^{i}_t := \left( |\sigma(i)|^2 +
\int_{\mathbf{R}^n} (e^{\sigma(i)\!^\top\!x} -1)^2\nu(dx)   + \left(
\sum_{j \in \mathcal{K} : j \neq i}  (e^{\rho^{i,j}}-1)^2\lambda^{i,j}\right) \right).
\end{equation*}
If $D$ is given by \eqref{eq:D}, the growth condition \eqref{eq:LG-D} holds, and the value function $v$ is sufficiently smooth, then by Theorem \ref{thm:risk-min-markov}, on the set $\set{C_{t-}=i}$, the component $\phi$ has the representation
\begin{align*}
\phi_t = (\widehat{G}^{i}_t)^{-1} &\left(|\sigma(i)|^2  \nabla v( t,S_{t-},i ) + \int_{\mathbf{R}^n}
(e^{\sigma(i)\!^\top\!x} -1) \frac{v( t,S_{t-}e^{\sigma(i)\!^\top\!x},i ) - v( t,S_{t-} ,i)}{S_{t-}} \nu(dx) \right.\\
&\quad+
\left.
  \sum_{j \in \mathcal{K} : j \neq i} (e^{\rho^{i,j}}-1)  \frac{(v( t,S_{t-}e^{\rho^{i,j}},j ) - v( t,S_{t-},i ) + \delta^{i,j}(t,S_{t-}))}{S_{t-}}  \lambda^{i,j}
\right),
\end{align*}
We can use Theorem \ref{thm:risk-min-markov} since \eqref{eq:sigma-2m} implies \eqref{eq:sup-S2}.
In the case of generalized exponential L\'{e}vy model described above, to find the value function $v$ we
may try to solve the system of PIDE with the corresponding terminal conditions (see Theorem \ref{thm:Fey-Kac-nowy}), which in this model takes the following form:
    \begin{align}\label{eq:gen-exp-levy-PIDE}
    & \partial_t v(t,s,c)
    +
    \nabla v(t,s,c) \mathfrak{r}(c) s
    + \frac{1}{2}
     s |\sigma(c)|^2 \nabla^2 v(t,s,c)
    \\ \nonumber
     & + \int_{\mathbf{R}^n}\left( v( t,se^{\sigma(c)\!^\top\!x},c)
    - v( t,s,c) - \nabla v(t,s,c)
    (e^{\sigma(c)\!^\top\!x}-1) \right)\nu(dx)
    \\
    \nonumber & + \sum_{c' \in \mathcal{K} \setminus c }
    \left( v( t, s e^{\rho^{c,c'}},c') - v(t,s,c) - \nabla v(t,s,c)
    (e^{\rho^{i,j}}-1)+ \delta^{c,c'}(t,s) \right)\lambda^{c,c'}
    \\ \nonumber
    &+ g(t,s,c) = 0,\\
    \nonumber
    & v(T,s,c) = h(s,c).
    \end{align}

In the case of $\rho^{c,c} =0$ and $\nu \equiv 0 $, the above system
of PIDE reduces to that  considered by Becherer and Schweizer
\cite{bs2005}, who gave sufficient conditions for existence of
bounded classical solution to this system of PIDE. Generally, the problem of existence of classical solutions to this PIDE is complex and it is well known that in some cases classical solutions to this kind of PDE's may not exist (see Norberg \cite{nor2005}).
However, under some strong assumptions on functions $h,g, \delta^{i,j}, \lambda^{i,j}$,
we can find a classical solution to PIDE system
\eqref{eq:gen-exp-levy-PIDE}.
    To see this we assume, for simplicity,  that $\mathfrak{r} \equiv 0 $. Let us  consider the function
    \begin{align}\label{eq:clas-sol}
        v(t,s,c) &:= \mathbb{E}\left(  h(S^{0,s,c}_{T-t}, C^{0,s,c}_{T-t}  )
        + \int_0^{T-t} n(S^{0,s,c}_{u-},C^{0,s,c}_{u-}) du \right),
    \end{align}
    where $n$ is defined by
    \[
        n(s,c) := g(s,c) + \sum_{c' \neq c } \delta^{c,c'}(s)
        \lambda^{c,c'},
    \]
    and $(S^{0,s,c}_{u},C^{0,s,c}_{u})$ is the  solution to SDE \eqref{eq:SDE-exp-levy} starting from $(s,c)$ at time $0$.
We claim that if  $h$ and $n$ are  $\mathcal{C}^{2}$--function  in $s$ with
bounded first and second derivatives, then $v$ given by \eqref{eq:clas-sol} is a classical
solution to PIDE system \eqref{eq:gen-exp-levy-PIDE}. First let us note using
Markov property of solution to our SDE   that
    \begin{align*}
     v(t,S^{0,s,c}_t,C^{0,s,c}_t) =   \mathbb{E}\left(  h(S^{0,s,c}_T, C^{0,s,c}_T ) + \int_t^T g(S^{0,s,c}_{u-},C^{0,s,c}_{u-}) du
     + \sum_{i,j \neq i} \int_t^T \delta^{i,j}(S^{0,s,c}_{u-}) d H^{i,j}_u
     \Big|
     \mathcal{F}_t
     \right) .
    \end{align*}

By Theorem \ref{thm:FK-nec} we know that $v$ solves  PIDE
\eqref{eq:gen-exp-levy-PIDE} provided that  the function  $v$ is in  $\mathcal{C}^{1,2} \cap \mathcal{I}$, \eqref{eq:sup-S2} and \eqref{eq:LG-D} hold.
 Assumption \eqref{eq:sigma-2m} implies, by Remark \ref{rem:mom-est}a, that \eqref{eq:sup-S2} holds for $m=1$. Moreover our assumptions on functions $h$ and $n$ implies  \eqref{eq:LG-D}.
So it is enough to show required smoothness of $v$.
 Let us notice that $S^{0,s,c}_{u}$ is
continuously differentiable with respect to $s$, and the first two
derivatives are given by
    \begin{align*}
    & D^s_u := \partial_s  S^{0,s,c}_{u} = \frac{S^{0,s,c}_{u}}{s}
    \\
     & =\exp\left( \int_{0}^t J(\sigma(C^{0,s,c}_{u-})) du + \int_{0}^t\sigma( C^{0,s,c}_{u-}) d Z_u
     + \sum_{i,j\in \mathcal{K}: j \neq i } \int_0^t  \rho^{i,j} d H^{i,j}_u\right),
    \\
   & \partial^2_s  S^{0,s,c}_{u} = 0 .
    \end{align*}
    Therefore,  $D^s_u$  is square integrable for every $u$.
By assumption,  $h$ and $n$ are function of $\mathcal{C}^{2}$ class in $s$ with
the  bounded first and second derivatives, so by the dominated
convergence theorem we have
    \begin{align*}
        \partial_s v(t,s,c) &= \mathbb{E}\left(  \partial_s h(S^{0,s,c}_{T-t}, C^{0,s,c}_{T-t}  )
        D^s_{T-t} + \int_0^{T-t} \partial_s  n(S^{0,s,c}_{u-},C^{0,s,c}_{u-}) D^s_{u} du \right),
         \\
        \partial^2_s v(t,s,c) &= \mathbb{E}\left(  \partial^2_s h(S^{0,s,c}_{T-t}, C^{0,s,c}_{T-t}  )
         (D^s_{T-t})^2 + \int_0^{T-t} \partial^2_s  n(S^{0,s,c}_{u-},C^{0,s,c}_{u-}) (D^s_{u})^2 du
         \right).
    \end{align*}
Boundedness of derivatives of $h$ and $n$, together with
independence of $D^s_u $ of $s$,     yield that these derivatives
are bounded functions in $s$.
 This implies, by Remark \ref{rem:mom-est}b, that $v \in \mathcal{I}$.
To get that $v$ is $C^1$ class in $t$
we first apply the It\^{o} lemma to function $h$ and obtain
    \begin{align*}
        v(t,s,c) = h(s,c) + \mathbb{E}\left(  \int_0^{T-t} (\mathcal{A}_e h + n)(S^{0,s,c}_{u},C^{0,s,c}_{u}) du
        \right),
    \end{align*}
    where $\mathcal{A}_e$ denotes the operator
    \begin{align}
    \nonumber
        \mathcal{A}_e h (s,c) :=
            & \frac{1}{2}
     s^2 |\sigma(c)|^2 \nabla^2 h(s,c) + \sum_{c' \in \mathcal{K} \setminus c }
     \left( h( s e^{\rho^{c,c'}},c') - h(s,c)\right)\lambda^{c,c'}
    \\ \nonumber
     & + \int_{\mathbf{R}^n}\left( h( se^{\sigma(c)\!^\top\!x},c)
    - h( s,c) - \nabla h(s,c)
    (e^{\sigma(c)\!^\top\!x}-1) \right)\nu(dx).
    \end{align}
 Since the process
    \begin{align*}
      f(u) =  (\mathcal{A}_e h + n)(S^{0,s,c}_{u}, C^{0,s,c}_{u}  )
    \end{align*}
is continuous in probability, using the Pratt theorem we see that
\begin{align*}
        v(t,s,c) = h(s,c) +  \int_0^{T-t} \mathbb{E}\left( (\mathcal{A}_e h +
        n)(S^{0,s,c}_{u},C^{0,s,c}_{u})\ \right)\  du,
    \end{align*}
so $v \in \mathcal{C}^{1,2} \cap \mathcal{I}$, which implies our claim.
{
\subsubsection{Semi-Markovian regime switching models}
Now we present how in our framework a feedback mechanism in jumps of $Y$ and intensity of jumps of $C$ gives extra flexibility in modeling. As an example we present how  semi-Markov switching processes can be embedded in our framework. Let us recall that a semi-Markov nature of $C$ is reflected in  the fact that the compensator  $\lambda_{i,j}$ of jumps from $i$ to $j$ depends on time that process $C$ spends in a current state after the last jump.
A semi-Markov regime switching process can be embedded in our framework by  considering the following SDE
\begin{align*}
    d S_t &= S_{t-}\!\!\left( r dt + \sigma(C_{t-}) d W_t  + \!\!\int_{\mathbf{R}}\!\!
    (e^{\sigma(C_{t-}) x } - 1 ) \widetilde{\Pi}( dx, dt)\right) \\
    d R_t&= dt - \sum_{i,j \in \mathcal{K} : j \neq i} R_{t-} \I_\set{ i} (C_{t-}) d
    N^{i,j}_t \\
    d C_t &= \sum_{i,j \in \mathcal{K} : j \neq i} (j-i) \I_\set{ i} (C_{t-}) d
    N^{i,j}_t ,
\end{align*}
where $W$ is a Wiener process, $\sigma(i) \geq 0$, $N^{i,j}$ are the point processes with
intensities $\lambda^{i,j}$ being functions of $R_{t-}$,
$\Pi(dx,dt)$ is a
Poisson random measure with intensity measure $\nu(dx)dt$ satisfying
\begin{align}\label{eq:SM-levy-exp}
    \int_{|x|>1}e^{2 \sigma(i) x} \nu(dx) <
    \infty
    \qquad
    \forall
    i \in \mathcal{K}.
\end{align}
Note that the coefficients of this SDE satisfy assumptions of Theorem 5.3 \cite{jaknie2012}, and if functions $\lambda^{i,j}$ satisfy assumptions of this theorem
 then there exists a solution unique in law.
Moreover, the  coordinate $C$ of solution $(S,R,C)$, is a semi-Markov chain
with the state space $\mathcal{K}$. Theorem \ref{thm:risk-min-spec} yields that the component $\phi$ of the risk-minimizing strategy
for a  dividend process $D$ with representation \eqref{eq:repr-H-P} is given, on the set $\set{C_{t-}=i}$, by
\begin{equation*}
\phi_t = \frac{\left( \sigma(i)\!^\top\!\delta_t + \int_{\mathbf{R}^n}
(e^{\sigma(i)\!^\top\!x} -1) J_t(x) \nu(dx)
\right)}{S_{t-}\widehat{G}^{i}_t},
\end{equation*}
where
\begin{equation*}
\widehat{G}^{i}_t := \left( |\sigma(i)|^2 +
\int_{\mathbf{R}^n} (e^{\sigma(i)\!^\top\!x} -1)^2\nu(dx)  \right).
\end{equation*}
If $D$ is given by \eqref{eq:D}, the growth condition \eqref{eq:LG-D} holds, and the value function $v$ is sufficiently smooth, then we can use  Theorem \ref{thm:risk-min-markov},  since \eqref{eq:SM-levy-exp} implies \eqref{eq:sup-S2}. So, we have, on the set $\set{C_{t-}=i}$,
\begin{align*}
\phi_t = & (\widehat{G}^{i}_t)^{-1} \Big(|\sigma(i)|^2  \nabla v( t,S_{t-},R_{t-},i )  \\ &   \quad +\int_{\mathbf{R}^n}
(e^{\sigma(i)\!^\top\!x} -1) \frac{v( t,S_{t-}e^{\sigma(i)\!^\top\!x},R_{t-},i ) - v( t,S_{t-}, R_{t-} ,i)}{S_{t-}}  \nu(dx) \Big).
\end{align*}
In the case of semi-Markovian regime switching exponential L\'{e}vy model described above, one of the way of finding the value function $v$
is to solve the corresponding system of PIDE with the  terminal conditions
(see Theorem \ref{thm:Fey-Kac-nowy}).



\section{Appendix}
In this appendix we prove the results which is outside the main scope of our paper, but we need it in the proof of Theorem \ref{prop:attain}. This result  connects the form of GKW decomposition of $  \int_0^T \frac{1}{B_u} d D_u $  with the existence  of replication strategy for a payment stream  $D$.
\begin{lem}\label{lem:replication}
The following conditions are equivalent:
\begin{enumerate}
  \item There exist a $\mathbb{P}$-admissible strategy $\varphi$ which replicates a payment stream  $D$.
  \item In the GKW decomposition \eqref{eq:GKW} of $V^*$   we have $L^X \equiv 0$.
\end{enumerate}
\end{lem}
\begin{proof}
$\Longleftarrow$
    Follows immediately from \cite[Lem. 16.17]{jaknie2011}.

$\Longrightarrow$
    Assume that $\varphi$ is a $\mathbb{P}$-admissible strategy which replicates $D$.
    By \cite[Prop. 16.14]{jaknie2011}  we have that
    \[
        V_t(\varphi) = B_t \mathbb{E} \left( \int_t^T \frac{1}{B_u} d D_u |
\mathcal{F}_t \right).
    \]
    Hence
    \[
        \frac{V_t(\varphi)}{B_t} - \int_0^t \frac{1}{B_u} d D_u =  \mathbb{E} \left( \int_0^T \frac{1}{B_u} d D_u |
\mathcal{F}_t \right).
    \]
    Using \cite[Lem. 16.7]{jaknie2011}  we obtain
    \[
        V_0(\varphi) + \int_0^t \phi_u^\top d S^{*}_u =  V^*_t.
    \]
\end{proof}


\begin{thebibliography}{10}

\bibitem{Albert1972}
A.~Albert
\newblock{ \it Regression and the Moore-Penrose Pseudoinverse},
\newblock{ Academic Press}, (1972)

\bibitem{apple2004}
D.~Applebaum
\newblock{ \it
L\'{e}vy Processes and Stochastic Calculus },
\newblock{ Cambridge University Press}, (2004).



\bibitem{bech2006}
D.~Becherer
\newblock{ \it
Bounded  solutions  to backward  SDE's  with  jumps for utility optimization  and indifference  hedging},
\newblock{ Annals of Applied Probability }
 Vol. 16, No. 4, (2006), 2027--2054.


 \bibitem{bs2005}  D.~Becherer and
M.~Schweizer. \newblock {\it Classical Solutions To
Reaction-Diffusion Systems For Hedging Problems With Interacting It\^o
And Point Processes },
\newblock {Annals of Applied Probability} Vol. 15, No. 2, (2005), pp.
1111--1144.


\bibitem{ch2005}
K.~Chourdakis
\newblock{ \it Switching L\'evy Models in Continuous Time: Finite Distributions and Option Pricing}
\newblock{  Centre for Computational Finance and Economic Agents (CCFEA) Working Paper. (2005)}

\bibitem{convol2005a}
R.~Cont and E.~ Voltchkova.
\newblock{ \it A Finite Difference Scheme For Option Pricing In
Jump Diffusion And Exponential L\`{e}vy Models }.
\newblock{ SIAM J. Numer. Anal. } Vol. 43, No. 4, pp. 1596–-1626.

\bibitem{convol2005b}
R.~Cont and E.~ Voltchkova.
\newblock{ \it Integro-differential equations for option prices in exponential L\'{e}vy models }.
\newblock{ Finance and Stochastics } Vol. 9 (3), July 2005, pp. 299 -- 325 .

\bibitem{convol2007}
R.~Cont, P.~Tankov and E.~ Voltchkova.
\newblock{ Hedging with options in presence of jumps }.
\newblock{ Stochastic Analysis and Applications: The Abel Symposium  2005 in honor of Kiyosi It\^o} ,  Springer 2007, pages 197--218.



\bibitem{ef1991}
R.J. Elliott and H. F\"{o}llmer
\newblock{ \it Orthogonal martingale representation }
\newblock{ Liber Amicorum for M. Zakai}, Academic Press (1991): 139--152.

\bibitem{ecs2005}
R.J. Elliott, L. Chan and T.K. Siu
\newblock{
Option pricing and Esscher transform under regime switching}
\newblock{ Annals of Finance}, Vol. 1, No. 4 (2005),423--432


\bibitem{folson1986} H.~F\"{o}llmer and  D.~Sondermann
\newblock{ \it Hedging of non-redundant contingent claims }, In
W. Hildenbrand and A. Mas-Collel(eds.), Contributions to Mathematical Economics, (1986),
205-223

\bibitem{folsch1991}  H.~F\"{o}llmer and  M.~Schweizer. \newblock {\it Hedging of contingent claims under incomplete information },
\newblock {In: Davies, M.  and  Elliot, R. (Eds.) Applied Stochastic Analysis in: Stochastic Monographs } Vol. 5, (1991), pp.
389--414.

\bibitem{gichskorIII}
I.I.~Gichman and A.V.~Skorochod,
\newblock{The Theory of Stochastic Processes III},
\newblock{\it Grundlehren der mathematischen Wissenschaften, 232, Springer-Verlag, Berlin 1979}.

\bibitem{hewangyan} Sh.~He, J.~Wang and J.~Yan
\newblock{ \it  Smimartingale theory and stochastic calculus }
\newblock{ CRC Press Inc (1992).}

\bibitem{hunt2010}
Hunt J.S.
\newblock{ \it A short note on continuous-time Markov and semi-Markov processes},
\newblock{ discussion paper (DP-1003), ISBA, UCL, (2010)}


\bibitem{jp1978} J. Jacod and P. Protter.
\newblock{ \it Quelques remarques sur un nouveau type d'\'{e}quations diff\'{e}rentielles stochastiques},
\newblock{ S\'{e}minaire de Probabilit\'{e}s XVI }. Lecture Notes in Math. 920, pp.
447–458. Springer, Berlin. (1982).

\bibitem{js1987}
J.  Jacod  and A. N. Shiryaev.
\newblock{\it Limit theorems for stochastic processes.},
\newblock{ Grundlehren der Mathematischen Wissenschaften, 288. Berlin etc.:
    Springer-Verlag. XVIII, 601 p. (1987). }
\bibitem{jaknie2010} J.~Jakubowski and M.~Nieweglowski.
\newblock {\it A class of F doubly stochastic Markov chains} ,
\newblock {Electronic journal of probability} Vol. 15 (2010), pages 1743--1771.


\bibitem{jaknie2011}
    J. Jakubowski and  M. Niew\k{e}g\l owski
    \newblock{\it Pricing and hedging of rating-sensitive claims modelled by F-doubly stochastic Markov chains.},
\newblock{  in: G. Di Nunno, B. Oksendal (eds.)  Advanced Mathematical Methods for Finance, Springer (2011), 417--454}

\bibitem{jaknie2012}
    J. Jakubowski and  M. Niew\k{e}g\l owski
    \newblock{\it Solutions to SDE's with Lévy noise in random environments}, arXiv:1305.4129.

\bibitem{kflr2012}
Y.S. Kim, F.J. Fabozzi, Z. Lin and S.T. Rachev
\newblock{\it Option pricing and hedging under a stochastic volatility Lévy process model},
\newblock{ Review of Dervatives Research }.
Volume 15, Number 1 (2012), 81--97.

\bibitem{kurtz2010}
T.G.~Kurtz
\newblock{ \it Equivalence of Stochastic Equations and Martingale Problems},
\newblock{Stochastic Analysis}  2010, 113--130.


\bibitem{mijpis2011}
    A. Mijatovic and  M. Pistorius
    \newblock{\it Exotic derivatives under stochastic volatility models with jumps.},
\newblock{  in: G. Di Nunno, B. Oksendal (eds.)  Advanced Mathematical Methods for Finance, Springer (2011), 455--508}

\bibitem{mol2001}
T.~M{\o}ller
\newblock{ \it Risk-minimizing hedging strategies for insurance payment processes.}
\newblock{ Finance and Stochastics 5(4), 2001, 419-446. }

\bibitem{nor2005}
R.~Norberg
\newblock{Anomalous PDEs in Markov chains: Domains of validity and numerical solutions}
\newblock{Finance Stochast. 9, (2005), 519–-537 }


\bibitem{prot2004}
P.E.~Protter.
\newblock {\em Stochastic integration and differential equations}.
\newblock Application of Mathematics. Springer-Verlag Berlin Heidelberg, New York, 2004


\bibitem{RW2000}
Rogers, L.C.G. and  Williams D.,
\newblock{ \it Diffusions, {M}arkov {P}rocesses and {M}artingales},
\newblock{ Cambridge University Press, 2000}


 \bibitem{sch1993}  M.~Schweizer. \newblock {\it Option hedging for semimartingales.},
\newblock {Stochastic processes  and their Applications } Vol. 37, (1993), pp.
339--363.

\bibitem{sch2001} M.~Schweizer
\newblock{ \it A Guided Tour through Quadratic Hedging Approaches }
\newblock{ in: E. Jouini, J. Cvitanic, M. Musiela (eds.)} , "Option Pricing, Interest Rates and Risk Management", Cambridge University Press, (2001), pp 538--574.

\bibitem{sy2009}
 T.K. Siu and H. Yang
\newblock{ \it  Option pricing when the regime-switching risk is priced }
\newblock{  Acta Mathematicae Applicatae Sinica }
Volume 25, Number 3 (2009), 369--388

\bibitem{tankov2011} P.~Tankov
\newblock{ \it Pricing and Hedging in Exponential Lévy Models: Review of Recent Results}
\newblock{ in: A.R. Carmona et al. (eds.), "Paris-Princeton Lectures on Mathematical Finance 2010",
Lecture Notes in Mathematics 2003, Springer-Verlag Berlin Heidelberg 2011}

\bibitem{yzz2006}
D.D. Yao, Q. Zhang and X.Y. Zhou
\newblock{\it A Regime-Switching Model for European Options}
\newblock{ Stochastic Processes, Optimization, and Control Theory: Applications in Financial Engineering, Queueing Networks, and Manufacturing Systems
A Volume in Honor of Suresh Sethi, International Series in Operations Research \& Management Science
Volume 94, 2006 }

\bibitem{yy2009}
F.L. Yuen and H. Yang,
\newblock{\it Option Pricing in a Jump-Diffusion Model with Regime-Switching}
\newblock{ASTIN Bulletin}, Volume 39, No. 2,  November 2009


%
%
%
%


 \end{thebibliography}
\end{document}